\documentclass[letterpaper,twocolumn,10pt]{article}
\usepackage{usenix-2020-09}

\usepackage{tikz}
\usepackage{amsmath}

\newcommand{\eat}[1]{}
\newcommand{\del}[1]{}

\usepackage{textcomp}
\usepackage{setspace}
\usepackage[title]{appendix}
\usepackage{ifthen}
\usetikzlibrary{positioning}
\usepackage{listings}
\usepackage{epsfig,endnotes}
\usepackage{xspace}
\usepackage[linesnumbered,ruled,vlined]{algorithm2e}
\usepackage{float}
\usepackage{graphicx}
\usepackage[labelformat=simple]{subcaption}
\usepackage{kbordermatrix}
\renewcommand{\kbldelim}{(}
\renewcommand{\kbrdelim}{)}
\usepackage{amsthm} 
\usepackage{amssymb}
\usepackage{pifont}
\usepackage{enumitem}
\usepackage{xspace}
\usepackage{caption}
\usepackage{booktabs}  
	
\usepackage{tabularx}
\usepackage{array}
\usepackage{arydshln}
\usepackage{multirow}
\usepackage{ragged2e}  
\usepackage{outlines}
\usepackage[subtle]{savetrees}
\usepackage{diagbox}

\newcolumntype{C}{>{\centering\arraybackslash}X}
\newcolumntype{R}{>{\RaggedRight\arraybackslash}X}
\newcolumntype{L}{>{\RaggedLeft\arraybackslash}X}

\usepackage{titlesec}
\titlespacing{\subsection}{0pt}{*0.8}{*0.8} 

\usepackage{pifont}

\usepackage[capitalize,noabbrev]{cleveref}
\crefname{equation}{Eq.}{Eqs.}
\Crefname{equation}{Eq.}{Eqs.}
\crefname{theorem}{Thm.}{Thms.}
\Crefname{theorem}{Thm.}{Thms.}
\crefname{figure}{Fig.}{Figs.}
\Crefname{figure}{Fig.}{Figs.}
\crefname{table}{Table}{Tables.}
\Crefname{table}{Table}{Tables.}

\crefformat{section}{\S#2#1#3}
\crefformat{subsection}{\S#2#1#3}

\Crefformat{section}{\S\S#2#1#3} 
\Crefformat{subsection}{\S\S#2#1#3} 

\newcommand{\T}[1]{\par\vspace{2pt plus 1pt minus 1pt}\noindent\textbf{#1}} 
    \newcommand{\Ts}[1]{\par\noindent\textit{#1}} 

\newcommand{\thm}[1]{Thm.~\ref{#1}} 

\newcommand{\fig}[1]{Fig.~\ref{#1}} 

\newcommand{\be}{\begin{equation}}
\newcommand{\ee}{\end{equation}}

\newcommand{\sysname}{CrossCheck\xspace}
\newcommand{\sys}{\sysname}
\newcommand{\name}{\sysname}
\newcommand{\geant}{G\'EANT\xspace}

\newcommand{\sr}[1]{\footnote{\color{cyan} SR: #1}}
\newcommand{\ri}[1]{\footnote{\color{blue} RI: #1}}
\newcommand{\ak}[1]{\footnote{\color{magenta} AK: #1}}
\newcommand{\ik}[1]{\footnote{\color{violet} IK: #1}}
\newcommand{\as}[1]{\footnote{\color{orange} AS: #1}}
\newcommand{\rjs}[1]{\footnote{\color{purple} RJS: #1}}
\newcommand{\IK}[1]{\ik{#1}}
\newcommand{\red}[1]{\noindent{\color{red} {\bf #1}}}
\newcommand{\topic}[1]{\textcolor{orange}{\bf #1}}

\newcommand{\cameraready}[1]{\noindent{\color{red} #1}}
\renewcommand{\cameraready}[1]{#1}

\renewcommand{\sr}[1]{}
\renewcommand{\ri}[1]{}
\renewcommand{\ak}[1]{}
\renewcommand{\as}[1]{}
\renewcommand{\ik}[1]{}
\renewcommand{\rjs}[1]{}
\renewcommand{\IK}[1]{}
\renewcommand{\red}[1]{}
\renewcommand{\topic}[1]{#1}

\usepackage{bm}
\usepackage[d]{esvect}

\newcommand{\NN}{\mathcal{N}}

\newcommand{\cutoff}{\Gamma}

\newcommand{\eg}{{\it e.g.,}\xspace}
\newcommand{\ie}{{\it i.e.,}\xspace}
\newcommand{\etc}{{\it etc.}\xspace}
\newcommand{\vs}{{\it vs.}\xspace}

\newcommand{\para}[1]{\left( #1 \right)}        

\newcommand{\sbrac}[1]{\left[ #1 \right]}
 
\newcommand{\bp}{\begin{proof}}
\newcommand{\bpo}{ \begin{proof}[Proof Outline] }
\newcommand{\ep}{\end{proof}}

\newtheorem{theorem}{Theorem}
 
\usepackage{pifont}

\usepackage{booktabs}  
\usepackage{multirow}

\begin{document}

\date{}

\title{CrossCheck: Input Validation for WAN Control Systems}

\newcommand{\aut}[2]{#1\texorpdfstring{$^{#2}$}{(#2)}} 
\author{
{\rm \aut{Alexander Krentsel}{1,2} \quad  \aut{Rishabh Iyer}{1} \quad \aut{Isaac Keslassy}{1,3} \quad \aut{Bharath Modhipalli}{2} 
}\\
{\rm \aut{Sylvia Ratnasamy}{1,2} \quad \aut{Anees Shaikh}{2} \quad \aut{Rob Shakir}{2}
}\\ \\
$^1$ \textit{UC Berkeley} \quad $^2$ \textit{Google} \quad $^3$ \textit{Technion}
}

\maketitle
\pagestyle{empty}

\thispagestyle{empty}

\begin{abstract}
We present CrossCheck, a system that validates inputs to the Software-Defined Networking (SDN) controller in a Wide Area Network (WAN). By detecting incorrect inputs -- often stemming from bugs in the SDN control \linebreak infrastructure -- CrossCheck alerts operators before they trigger network outages.  

Our analysis at a large-scale WAN operator identifies invalid inputs as a leading cause of major outages, and we show how CrossCheck would have prevented those incidents. We deployed CrossCheck as a shadow validation system for four weeks in a production WAN, during which it accurately detected the single incident of invalid inputs that occurred while sustaining a 0\% false positive rate under normal operation, hence imposing little additional burden on operators. In addition, we show through simulation that CrossCheck reliably detects a wide range of invalid inputs (e.g., detecting demand perturbations as small as 5\% with 100\% accuracy) and maintains a near-zero false positive rate for realistic levels of noisy, missing, or buggy telemetry data (e.g., sustaining zero false positives with up to 30\% of corrupted telemetry data).  
\end{abstract}

\section{Introduction}
\label{sec:introduction}

\topic{Modern society is increasingly reliant on the Internet, yet state-of-the-art networks continue to exhibit regular outages.}
Major outages have been reported by virtually every large network operator~\cite{aws-outage, azure-outage, gcp-outage, ibm-outage, meta-outage},
and a recent study shows that, despite sustained effort and investment, the frequency of such outages is not declining~\cite{dsdn}. 

\topic{To better understand these outages, we studied the postmortem reports for all major outages in a large cloud WAN over a five-year period (2019–2024).}
Like many WANs operated by large cloud providers~\cite{swan,ebb,b4}, this network employs an SDN-based control architecture in which a logically centralized software controller is responsible for routing and traffic engineering decisions.
As such, we believe the findings from our study are broadly representative.

\topic{The most common root cause across the outages we analyzed---accounting for over one-third of all cases---was \emph{incorrect inputs} to the SDN controller, i.e., inputs that \emph{did not accurately reflect the current state of the network}.}
For example, in some outages, the SDN controller received an incomplete view of the current traffic demand, while in others it received an incorrect view of the current network topology. 
As expected, decisions based on incorrect inputs resulted in undesirable outcomes, including sub-optimal routes, congestion, link overloads, and packet loss.
Informal conversations with practitioners at other cloud providers confirm that such incorrect inputs are not specific to the network we studied, but are also a common cause of outages elsewhere.

\topic{Why do incorrect inputs occur? The answer lies in the complexity of production WANs.}
State-of-the-art deployments incorporate control infrastructures spanning dozens of services and millions of lines of code, each subject to frequent updates.
In addition, they rely on routers that are typically sourced from multiple vendors and comprise complex hardware and software stacks.
Together, this complexity creates a large surface area for bugs, making it nearly inevitable that some controller inputs will be missing, stale, or incorrect.

\topic{Given that incorrect inputs are inevitable,  we ask:} 
\emph{Can we build a system that detects when inputs to the SDN controller deviate from the current network state?}
Here, inputs refer to higher-level, aggregated information such as traffic demand matrices or topology views, while current state is defined in terms of low-level dataplane signals exposed by routers, such as interface byte counters, link status indicators, and forwarding entries.
Such a system would allow operators to detect input errors before their effects manifest, thereby avoiding a large class of outages.

\topic{To be useful, such an input validation system must satisfy two key requirements.}
First, it must continuously validate inputs in real time, at the timescales at which SDN control decisions are made. 
This requirement is motivated by our analysis of postmortem reports, which revealed that incorrect inputs typically arise not because they are \cameraready{syntactically} invalid \cameraready{or impossible}, but because they are inconsistent with the \emph{current} state of the network. 
Second, the system must maintain a near-zero False Positive Rate (FPR), \ie not raise an alert when the inputs are correct,
while achieving a high True Positive Rate (TPR), \ie reliably raise an alert when the inputs are incorrect.  
We prioritize minimizing the FPR because, in a typical WAN, inputs are correct for the vast majority of time. 
Thus, even a modest FPR (e.g., 1\%) would produce an unacceptable number of spurious alerts, as controller decisions are executed frequently (typically on the order of minutes). 
This would undermine operator trust and limit adoption. 


\topic{Meeting these requirements in a production-scale WAN presents several challenges.}
First, measurements in real-world networks are inherently noisy, and the system must reliably distinguish benign noise from input errors to avoid false positives.
Second, as noted earlier, even the low-level router signals used to represent current state may be incorrect due to 
router bugs.
This introduces a seemingly recursive problem since the system must validate inputs against reference signals that may themselves be incorrect.
Finally, implementing such a validation system in practice is non-trivial.
The system must operate in real time, remain decoupled from the control plane to avoid shared failure modes, and be sufficiently simple to minimize the risk of bugs in the validator itself.


\topic{We present \textit{\sys}, an input validation system that can reliably detect incorrect inputs to an SDN controller even in the presence of noisy or faulty router signals.}
\sys builds on two key insights.
First, flow conservation in a network yields multiple, redundant measurements that should remain \emph{consistent} during normal operation. For example, the \texttt{bytes\_out} counter on one end of a link should match the \texttt{bytes\_in} on the other end, and the traffic demand reported between a WAN ingress router and an egress one should be reflected in the interface counters of all intermediate routers.
Second, any inconsistencies between these measurements reveal both the \emph{presence} and the \emph{nature} of these errors. 
Specifically, faulty or noisy router signals appear as \emph{local} anomalies, such as a few counters reporting inconsistent values; whereas incorrect controller inputs, such as an incorrect demand matrix, produce \emph{global} inconsistencies visible across all routers along the affected paths.
\sys leverages this asymmetry to distinguish incorrect inputs from noisy or faulty router signals, which enables it to achieve a near-zero FPR while preserving a high TPR.


\topic{To reduce the likelihood of bugs in the validation process, \sys adopts an architecture that is decoupled from the general SDN control infrastructure.}
In \sys, inputs to the controller and router signals are continuously streamed into dedicated databases, 
and the validation logic is implemented as a simple, stateless process that meets the timing requirements.
Unlike current SDN control infrastructures, which are optimized for high performance and availability, {\sys}’s architecture is deliberately \emph{lean}, thereby reducing the likelihood of new or correlated bugs.

\topic{We evaluated \sys by running it as a shadow system for 4 weeks in the production WAN we analyzed, and by stress testing it in simulation across multiple topologies and error scenarios.}
In the production \cameraready{shadow} deployment, \sys incurred a $0$\% FPR and correctly detected the only incorrect input that is known to have occurred during the deployment period.
In simulation, \sys reliably detected incorrect inputs, and was able to identify \emph{all} scenarios in which demand estimates were perturbed by $5$\% or more. 
\sys also proved resilient to faulty telemetry, maintaining an FPR of $0$\% even when up to $30$\% of router signals were missing or corrupted.
Finally, and perhaps most importantly, {\sys}’s accuracy improves exponentially with network size, since larger networks provide more interdependent signals, which suggests that it is especially suited for production WANs. \cameraready{Because \sys only relies on fundamental network invariants, we also argue that it is more amenable to generalization.}

 




\eat{ 


\footnote{This paper extends a previous workshop publication~\cite{hotnets}. It provides more findings on major outages at a large WAN, a radically different solution (\name), a centralized streaming database implementation, and new evaluations that compare \name against the workshop algorithm (Hodor).}
Modern society is deeply and increasingly dependent on the communication infrastructure that underlies the Internet. Even critical infrastructures such as water, power, and transportation now rely on the Internet for automation and control. Availability is thus a network operator’s highest priority. Yet, despite considerable effort and investment, state-of-the-art networks continue to exhibit regular outages. 
For example, recent papers report that the frequency of major outages in large cloud provider WANs has remained roughly constant over the years~\cite{dsdn, meta-outage, aws-outage, azure-outage, gcp-outage, ibm-outage} and recent outages in carrier networks have led to widespread disruptions to daily life~\cite{att-1,others}.

What more can we do to avoid such outages? To answer this question, we analyzed the root cause of the \rjs{Is there an inconsistency here? Previouly we say that this is the last 5 years of outages.} 40 largest network outages
in a large cloud provider's SDN WAN over the last five years. 
Our analysis revealed that the single largest contributing root cause, with \emph{over one third} of these major outages, is one 
that has received relatively little attention within the research community: \emph{incorrect inputs to the SDN controller} that determines how traffic is routed over their network. 
In all these outages, the SDN controller itself operates correctly, but is compromised because it receives inputs that do not accurately reflect the current network state. 
For example, in several of the outages we analyzed, the controller received an incomplete view of the current traffic demand, leading to sub-optimal routes, congestion, and packet drops. 
In other outages, the controller received an incorrect view of the topology, which caused it to overload the links it believed to be operational. 

This data is troubling not only due to the large fraction of outages, but also because undetected wrong inputs evade existing techniques.
Techniques such as 
testing~\cite{netcastle}, emulation~\cite{crystalnet, kne}, and formal verification~\cite{netverify.fun, batfish, hsa, netkat, bagpipe}
validate a system's output \emph{assuming correct inputs}. As such, the guarantees they provide are rendered meaningless when the inputs to the system themselves are incorrect. 

Likewise, our discussions with operators reveal that they perform sanity checks on the inputs to their SDN controller. 
Operators craft static checks to prevent either \emph{impossible} input values, 
or \emph{unlikely} inputs that do not fit static heuristics based on historical patterns.

Unfortunately, static checks have three fundamental issues. (1)~Our analysis has revealed that incorrect inputs are most often neither impossible nor unlikely. Thus, these static checks would not catch most of them. Simply, incorrect inputs do not reflect the \emph{current} state of the network. 
(2)~Heuristics about unlikely inputs are also dangerous, since they can result in false positive scenarios where atypical inputs are discarded despite correctly reflecting the state of the network at that time; \eg in a disaster scenario that impacts a large number of routers. 
(3)~Finally, ad-hoc static checks for unlikely inputs are also hard to manage, as their heuristic rules keep evolving.

Why do incorrect inputs occur? 
We address this in greater detail in \S\ref{sec:motivation} but, briefly, the answer lies in the complexity of production WANs. 
For one, the control infrastructure 
that aggregates these inputs 
is complex, and spans dozens of services with millions of lines of code subject to frequent updates, making bugs unavoidable~\cite{blastshield,dsdn,govindan2016evolveordie}. 
Additionally, incorrect inputs may result from faulty router-level telemetry signals that arise due to bugs in router hardware and software, both of which are complex in terms of lines of code and often blackboxes to network operators.


The goal of this paper is to address the problem of input validation in WAN control systems, and design a system that can identify for the WAN operator when inputs are wrong, with few false positives and false negatives. 
%

We start by presenting our analysis of recent major outages. We show that input validation is an important problem that remains unsolved. We also explain why we want to have a particularly low false positive rate (FPR), \ie a low rate of false alarms. Given that inputs are correct most of the time, a high FPR can result in alerting the operator dozens of time per day for no reason, which would be unacceptable. A lower FPR may come at the expense of a higher false negative rate (FNR), \ie a higher probability of missing incorrect inputs and causing wrong decisions for the customer flows. 

\IK{Sylvia: does this start of the paragraph address your comments? Should we add how inputs and/or telemetry can get wrong? What else is missing?}\ak{the flow conservation properties still feel like they come a bit out of nowhere here}
We then focus on two types of data collected by WAN operators: (1)~high-level WAN controller's inputs constructed from both in-network and out-of-network signals, such as the demand matrix, and (2)~low-level network state information collected from the WAN routers, such as link-counter telemetry from which we deduce the link loads, and local path information from which we deduce the global paths. 
We argue that input validation must be based on \emph{dynamic invariants} that rely on inherent flow-conservation properties of the network. 
These invariants ensure that the \emph{current} network state reflects the WAN controller's inputs. For example, if we double all the demand matrix rates, we would expect the link loads within the network to double as well. 
These invariants also provide a self-consistency for the network state values. For example, the ingress rate at a link should always equal its egress rate. 

Given these invariants, there are still many challenges to building an input validation system.
A first challenge is that the equalities in these invariants do not hold exactly due to some noise. In particular, low-level counters are collected at routers, and each router starts with an arbitrary time-shift for its periodic collection windows. \ri{Can it really be arbitrary?}\IK{I think it is. Alex, please confirm?}\ak{arbitrary in the window, yes, the stream times aren't synced with one another, so it's periodic starting at a random time} In addition, packet losses and propagation delays induce additional noise.
To address this challenge, we use real-world WAN data to define a small noise threshold such that invariants will hold within the threshold with high probability.

We want to use the low-level counters to confirm or disprove the correctness of the high-level inputs. 
\sr{We might want to say something about the relationship between the high-level input and low-level info.} 
\IK{1. Should we add the formula Pd=c? Or is the example with doubling the rates sufficient? (see two paragraphs earlier) 2. Is it OK to focus on the demand matrix only as input? It simplifies the story, but then we miss the topology part.}
However, a second challenge is that the low-level counters may also be wrong. If our demand-matrix input is correct, but our erroneous low-level link loads do not match it, we will keep raising the alarm and suffer from a high FPR.
To address this challenge, we leverage the fact that we have many low-level counters that come from independent sources. 
We design a Democratic Trust Propagation (DTP) algorithm for input validation. To obtain the load at each link, \name relies on a \textit{majority vote} among several independent sources. This majority vote helps remove extreme values that may be due to wrong data, and keep the most likely load values. 
\sr{I'm worried the above para is throwing a compressed version of the algorithm at the reader without enough context and it's just going to be confusing - e.g., we're talking about load values but haven't said yet that we are (if we are) focusing only on demand, haven't explained why invariants are relevant (vs. just checking demands against current load), high-level vs. low-level inputs, etc.}
\IK{Sylvia: Please check this second iteration. What remaining concerns should be addressed?}

A third challenge is that input corruption can be correlated. For example, a router bug could zero all the link counters of the router, and a local majority vote could result in the belief that all the local link loads are zero. 

To mitigate local data corruptions, \name introduces two tools. It uses the information from neighbor routers, and it iteratively propagates \textit{global} information as in gossip algorithms. \name only settles locally when it is confident in the information because of consensus among data sources, while it tries to keep polling non-local counters when there is no consensus.

\IK{Alex: Let's also sell the fact that you found an issue with the topology in a production WAN... I'll let you phrase it correctly.}

Building on the above insights, we discuss the space of possible approaches, and implement a prototype for a  centralized off-the-shelf streaming database that could readily scale to WANs with thousands of routers. \sr{fix this}

Our evaluations suggest that \name could have averted the majority of the outages that stem from incorrect inputs in our dataset.

\sr{let's discuss live but I think we need to say something about why our proposal isn't just redoing what existing control infrastructures do. I.e., the systems arch piece. Feels like a significant gap in our reasoning otherwise.}

} 
\section{Background and Motivation}
\label{sec:background}

Over the past two decades, both traditional ISPs and modern hyperscalers have adopted SDN-based control architectures to manage their WANs~\cite{swan,ebb,b4,road-to-sdn}. 
At the core of these SDN systems is a logically centralized controller that runs a traffic engineering (TE) algorithm to compute capacity-aware paths~\cite{swan,b4,b4-after,ebb}. However, implementing TE in a global WAN entails far more than running a software process on a single datacenter server. In practice, it requires a variety of software services deployed across a global footprint of servers~\cite{swan,b4}. 
\cameraready{
These servers run a variety of services for topology discovery, demand estimation, switch programming, configuring policy, software upgrades, and more. Some of these services run on standard data center servers managed by a cluster management system~\cite{borg, twine, mesos, kubernetes} while edge controllers may run on special servers that are co-located with routers and run a specialized orchestration platform~\cite{orion,onos,dsdn}.

In addition, because each of these components — hardware and software alike — is on the critical path for high availability, they are engineered accordingly with  redundancy, replication, consensus, \etc, becoming fairly complex systems in their own right. For example, both edge and central controllers are commonly replicated across distinct hardware and geo-diverse data centers, running Paxos for consistency and failover~\cite{evolve-or-die}. 
Taken together, these systems represent dozens of non-trivial microservices and millions of lines of code. Prior work has highlighted the complexity of this control infrastructure and the inevitability of bugs and errors within the same~\cite{dsdn, blastshield,evolve-or-die,netcastle}. 
}

As noted earlier, we obtained sanitized diagnostic data from a major cloud provider covering all major outages over the last five years. Our analysis shows that more than one-third of these outages stemmed from incorrect inputs to the TE controller. To better understand how such errors occur, we examined individual outages in detail. In what follows, we first review the inputs to a TE controller, and then summarize our findings on how incorrect inputs arise and why standard best practices for input validation often fall short.

\subsection{TE Inputs}
In an SDN WAN, the inputs to the TE controller consist of two key pieces of information: (i)~{\bf traffic demand}, which is a matrix $D$, where $D_{ij}$ denotes the aggregate rate of traffic entering ingress router $i$ and destined for egress router $j$~\cite{rough2013-traffic-matrix-primer}; and (ii)~{\bf topology}, which includes physical connectivity (\ie which links are present) and link capacity (since partial cuts on bundled links can result in reduced but non-zero capacity). 

These inputs are typically computed from a few different data sources. In multiple networks that we are aware of, $D$ is computed from measurements at \textit{end hosts}~\cite{google2015bwe, swan}, while topology information is pieced together using network  information in the \textit{control plane} (\eg topology models~\cite{malt}, routers marked as undergoing maintenance, \etc) as well as link status signals read from individual \textit{routers}. Information from these various sources is then  aggregated as it travels from hosts and routers through the control hierarchy before being passed up to the TE controller. \cameraready{}

\subsection{How Do Incorrect Inputs Occur?}
Our analysis revealed that incorrect inputs can occur in either of the above inputs, and generally arise due to the three reasons we discuss below. 

\T{1) Incorrect inputs from external sources (e.g., end hosts).}
As mentioned, inputs such as demand may be collected from sources \emph{outside} of the network (\ie not from the routers). 
This can lead to scenarios where the SDN controller receives an incorrect view of demand, despite everything in the network working correctly. 
For example, in one scenario in our dataset, a new rollout of the demand instrumentation system introduced a bug that incorrectly aggregated demand at the end hosts. This caused the SDN controller to receive a partial view of the demand, which led to severe congestion and a major outage. 
In another major outage, the demand instrumentation service correctly measured demand, but this traffic was incorrectly throttled at the end hosts, causing the measured demand to differ from the traffic that was allowed onto the network, and the TE  controller to make suboptimal path selection decisions. 

\T{2) Incorrect router signals.}
Telemetry data provided by routers is one source of information that goes into computing the high-level topology input. In addition, we are interested in using router signals for validation, as representing ground truth for the network's current state.
However, router telemetry is itself not a flawless signal. 
Modern routers consist of complex hardware and millions of lines of code, both of which provide a large surface area for bugs.
For example, we observed a scenario in which a bug in the router operating system  caused certain telemetry messages to be duplicated, with one of the two messages reporting (at random) that the number of packets received on the router's interfaces was zero. These messages led the control plane to interpret these interfaces as faulty, incorrectly removing them from the topology input to the TE controller and leading to unnecessary congestion in the network. 
Additional examples included router bugs that led to malformed telemetry responses, changes in telemetry format (\eg from \texttt{string} to \texttt{int}), delayed telemetry reporting, and incorrect QoS marking on telemetry packets, all of which led to incorrect or missing router signals. 

\T{3) Bugs in the control plane infrastructure.} Bugs at any point in the processing infrastructure can cause correct data, whether host measurements or router signals, to be mutated or delayed. For example, in one outage, a new rollout of the topology instrumentation service introduced a bug that did not wait for all routers to provide their link statuses before stitching together the topology, thus providing the SDN controller with a partial view of the topology. 
In a second outage, a bug in a different instrumentation service caused it to misreport the liveness of particular links, leading once again to a partial topology view. 

\subsection{How Are Incorrect Inputs Detected Today?} 
Our discussions with operators indicate that, as \cameraready{one may expect}, they perform sanity checks on inputs \cameraready{to try to catch incorrect inputs}. 
However, these checks are typically \emph{static} and designed primarily to prevent impossible values, \ie inputs that cannot occur, such as topologies claiming more nodes than actually exist in the network.
Operators also apply static checks to flag unlikely inputs, relying on heuristics derived from historical values, past outages, and other experience. Unfortunately, these checks are often ad-hoc, defined in response to specific incidents rather than through a systematic process that considers inputs and network state holistically, as we aim to do. 
Operators further noted that these checks are difficult to maintain: they accumulate over time and must be updated continually to reflect evolving practices.
More critically, such heuristics are risky, producing false positives where atypical but valid inputs are discarded---for example, during a disaster that affects many routers simultaneously. Ultimately, the limitations of static checks are evident in the fact that they failed to detect the incorrect inputs that led to the outages we analyzed. For example, despite the existence of a static validation check that no single metro region was missing all routers, an outage was caused by a topology input bug where a large portion (but not all) of routers were dropped from many metros. 

\subsection{Bad Input Causes a Bad Day}

\cameraready{We walk through one situation in detail to illustrate how bad-input-caused outages arise despite the best intentions of network operators. In one outage we observed~\cite{incident-report-buggy-input}, an update was rolled out to a subset of the regional jobs that read interface telemetry directly from routers in the network, and aggregate it into abstract views of the topology. This update introduced a race condition bug, which when triggered caused the jobs to not properly wait for responses from all routers before constructing and passing up an abstract connectivity graph from the routers in their region. As a result, these connectivity graphs appeared to be missing a significant portion of the actually available capacity. 

The top-most abstract aggregator job received these partially incomplete sub-aggregations, and stitched them to produce a final global abstract topology missing roughly a third of actual available capacity, which was then given to the SDN controller to use for pathing and traffic placement.

The SDN controller produced paths using the available capacity in the abstract topology, fitting as much demand as possible but was unable to fit all demand because of the lack of capacity. This led to traffic throttling and congestion.

The SDN controller did have some static checks in place to catch badly malformed inputs; for example, checking that the provided topology was not empty, or that no single region was empty. These did not prevent the input from proceeding because the topology was not empty, and all regions had at least some capacity.

We note that in this scenario, the SDN controller's solver produced correct results given its inputs, i.e., the paths produced by the SDN controller were the best possible paths for a topology missing the capacity stated in the abstract topology. Rather, the issue was that the input given to the controller did not accurately reflect reality. 
}

\medskip
\textbf{In summary}, incorrect inputs to the SDN controller arise for a variety of reasons, and the current practice of using static sanity checks is both insufficient and risky. Given the scale and complexity of production networks, eliminating these bugs entirely via extensive testing or formal verification is also intractable~\cite{blastshield}. 
Motivated by this state of affairs, we propose \emph{input validation} as a complementary strategy that acts at runtime to contain the impact of such bugs.
\section{\sys Overview} 
\label{sec:overview}

We now present an overview of \sys, a system that continuously validates inputs to the SDN controller against the current \textit{network state}, which is defined in terms of low-level dataplane signals exposed by routers.
We begin by describing the overall architecture of \sys (\cref{sec:goals}), then introduce the specific router signals it collects and explain why they are suitable for validation (\cref{sec:signals}).
Finally, we provide intuition for how \sys exploits the inherent redundancy in network behavior to construct a reliable view of the current state, even in the presence of noisy or faulty signals (\cref{sec:invariants}).
We describe the design of \sys in detail in \cref{sec:design}.

\subsection{System Architecture}
\label{sec:goals}

\topic{\cref{fig:system-design} presents an architectural overview of \sys, which validates SDN controller inputs in three stages.} 
First, in the 
\ding{192} \emph{collection} step, router signals and controller inputs are continuously streamed into a centralized backend database.  
Second, in the \ding{193} \emph{repair} step, \sys constructs a reliable network-wide view of the current state from these collected signals.  
Third, in the \ding{194} \emph{validation} step, \sys checks whether the controller inputs are consistent with this reconstructed state and classifies them as either \emph{correct} or \emph{incorrect}.  
We adopt this binary decision model for simplicity, though \sys could be easily extended to \cameraready{additionally} abstain \cameraready{if} it detects that too many router signals are missing or corrupt for it to reach a confident verdict.

\topic{\sys's architecture is motivated by two considerations.}
First, decoupling \sys from the SDN control plane and streaming data into a dedicated backend allows the validation logic to remain simple, and isolated from shared failure modes.  
Second, repairing the network state holistically, enables \sys to exploit the inherent redundancy in router signals (\eg due to flow conservation across the two ends of a link or at a router) to identify and tolerate faulty signals and measurement noise.
Together, these choices enable \sys to maintain a near-zero False Positive Rate (FPR), \ie avoid flagging correct inputs, while achieving a high True Positive Rate (TPR), \ie reliably detecting incorrect ones.

\begin{figure}
    \centering
    \includegraphics[width=1\linewidth]{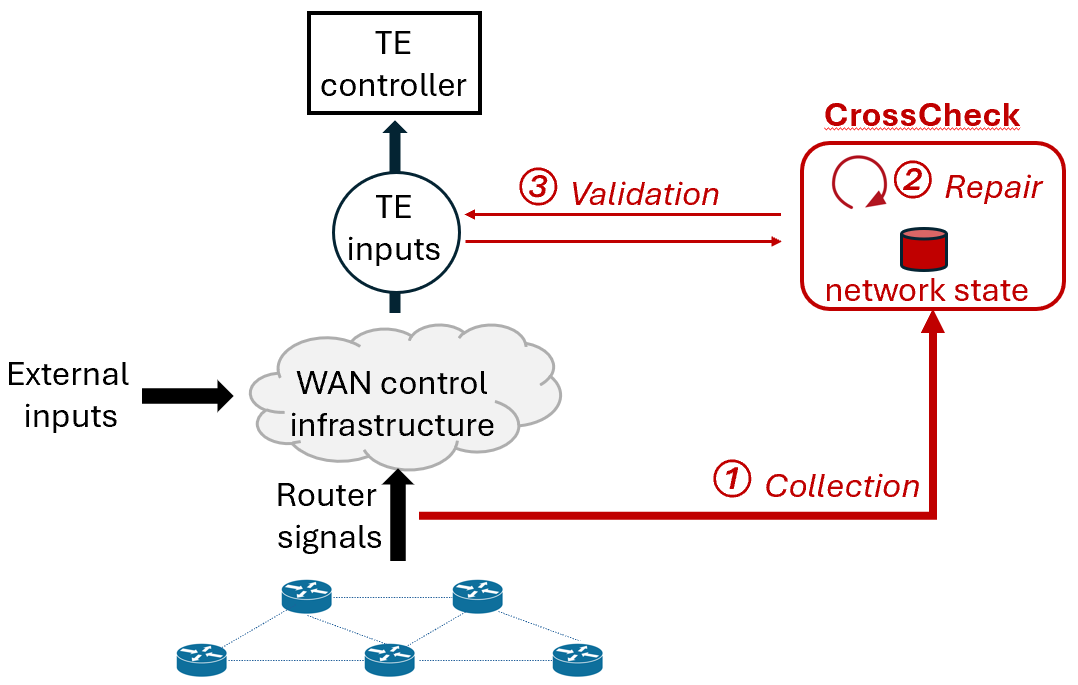}
    \caption{\sysname high-level system design.
    }
    \label{fig:system-design}
\end{figure}

\eat{
    \begin{figure}
        \centering
        \includegraphics[width=\linewidth]{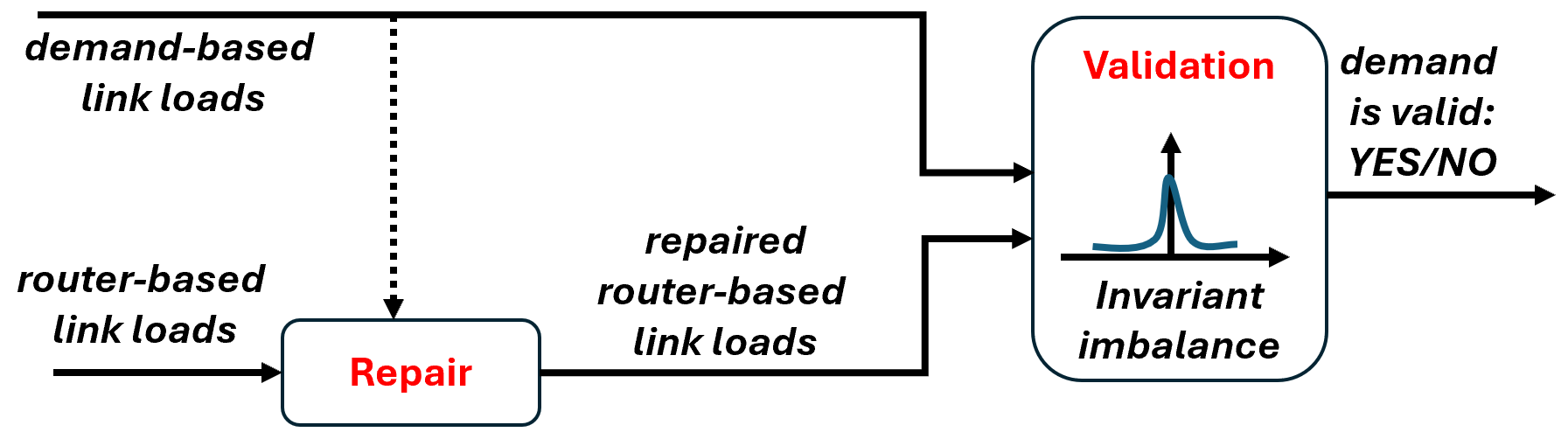}
        \caption{\sysname overview. \IK{I killed DTP and removed the DBs+processing, can bring back if needed}
        }
        \label{fig:overview}
    \end{figure}
}

\subsection{Collected Router Signals}
\label{sec:signals}

\topic{\sys collects three types of router signals.} 
We chose these signals because they are directly relevant to the inputs being validated (namely, demand and topology), and because they are consistently available across router platforms via standardized APIs~\cite{gnmi,gribi}.
\cref{tab:signals} summarizes the collected signals, which we explain below.

\begin{table}
\centering
\resizebox{\linewidth}{!}{
    \begin{tabular}{@{}l l c c@{}}
    \toprule
    \multirow{2}{*}{\bfseries Type} & \multicolumn{2}{c}{\bfseries Signal} & \multirow{2}{*}{\bfseries Notation} \\
    \cmidrule(lr){2-3}
    & \bfseries Name & \bfseries Location & \\
    \midrule
    \multirow{4}{*}{Link status indicators} & \multirow{2}{*}{Physical status} & egress & $l^X_{\mathit{phy}}$\\
    & & ingress &  $l^Y_{\mathit{phy}}$\\
    & \multirow{2}{*}{Link-layer status} & egress &  $l^X_{\mathit{link}}$\\
    & & ingress &  $l^Y_{\mathit{link}}$\\
    \midrule
    \multirow{2}{*}{Link counters} &  \multirow{2}{*}{Counters} & transmit & $l^X_{out}$ \\
    & & receive & $l^Y_{in}$\\
    \midrule
    Forwarding entries & Entries & router $X$ & $F_X (\to l_{demand})$ \\
    \bottomrule
    \end{tabular}
}
\caption{Collected router signals and their notations.}
\label{tab:signals}
\end{table}


\T{(1) Link status indicators.}
For each directed link $l$ from router $X$ to $Y$, \sys collects four link status indicators. 
First, the \emph{physical status} of the link at each router, denoted $l^X_{\mathit{phy}}$ and $l^Y_{\mathit{phy}}$, which reflects physical-layer conditions such as the detection of optical signals. 
Second, the \emph{link-layer status} at each router, denoted $l^X_{\mathit{link}}$ and $l^Y_{\mathit{link}}$, which indicates whether each router considers the link active, based on the successful exchange of heartbeat messages with the remote endpoint. 
\sys does not require additional heartbeat traffic for these signals; instead, it reuses the link status computed by existing protocols such as BFD~\cite{rfc5880bfd, rfc7130bfd-lag}, which are already deployed on modern routers.

\T{(2) Link counters.} 
\sys collects hardware counters from each router interface that track the cumulative number of bytes sent and received.
For a directed link from $X$ to $Y$, \sys uses the transmit counter at $X$ ($l^X_{out}$) and the receive counter at $Y$ ($l^Y_{in}$).
To compute traffic rates, \sys samples these counters periodically and derives per-interval rates from the difference in values and timestamps. For simplicity, we continue to refer to these as $l^X_{out}$ and $l^Y_{in}$.

\T{(3) Forwarding entries.}
Finally, \sys collects the forwarding table $F_X$ from each router $X$.
At ingress routers, $F_X$ specifies how incoming traffic is encapsulated into tunnels (via encapsulation rules), and at transit routers, it determines how each tunnel is forwarded through the network.
By combining forwarding entries across routers, \sys reconstructs the path of each tunnel and estimates the load contributed by each demand $D_{I,J}$ on every link. 
We denote this estimated load on link $l$ as $l_{demand}$.

\smallskip
We highlight two important properties of the signals described above.
First, among the seven link signals
only two---namely the physical statuses $l^X_{\mathit{phy}}$ and $l^Y_{\mathit{phy}}$---are used to compute the inputs to the SDN controller (specifically topology), and only one---namely the estimated load $l_{\mathit{demand}}$---is dependent on the inputs to the SDN controller, since $l_{\mathit{demand}}$ is derived using the demand input. 
The remaining four signals are thus completely independent of the controller inputs, hence offering a distinct view of the network state that can be used to validate those inputs.
That said, \sys does not rely solely on these independent signals, since excluding the input-dependent signals would increase vulnerability to faults in the others (\S\ref{sec:eval}).

Second, each of the signals is collected from a different component within a router, which reduces (though does not eliminate) the possibility that they are all simultaneously buggy.
In our analysis of all major outages logged over a five-year period, we found zero occurrences of simultaneous bugs across even two of the above three classes. 
\sys leverages this independence when validating the network topology, which we describe in \cref{sec:design-topology}. 
    
\subsection{Network Invariants}
\label{sec:invariants}

If the router signals collected by \sys were always accurate, validating SDN controller inputs would be straightforward.
For instance, one could validate demand by comparing the router-measured load on a link, $l_{\mathit{router}} = (l^X_{\mathit{out}} + l^Y_{\mathit{in}})/2$, with the load inferred from the demand input, $l_{\mathit{demand}}$.
Similarly, one could validate topology by checking whether the link-layer status indicators $l^X_{\mathit{link}}$ and $l^Y_{\mathit{link}}$ agree with the controller’s view of the link’s state.

In practice, however, router signals are often noisy, missing, or \cameraready{even} incorrect (\cref{sec:background}).
Hence, \sys first detects and corrects such faults in a process we call \emph{repair}.
This repair step enables \sys to reconstruct an accurate view of the network’s state, which it then uses to validate controller inputs while maintaining a low false positive rate.

\sys's repair process leverages the observation that router signals are \emph{strongly correlated} across the network. As a result, faults in one signal often produce inconsistencies with others.

Concretely, \sys leverages four network invariants that must hold in any correct network to repair faulty router signals.
We now define these invariants for a directed link $l$ from router $X$ to router $Y$, and describe how \sys performs the repair in the next section.

\T{(1) Link invariants.} 
These reflect the requirement that (i)~both ends of a link must agree on its operational status, in both the physical and link layers: 
\begin{equation}
l^X_{phy} = l^Y_{phy} = l^X_{link} = l^Y_{link},
\label{eq:link-status}
\end{equation}
and (ii)~the link must preserve flow conservation:
\begin{equation}
l^X_{out} = l^Y_{in}
\label{eq:link-counters}
\end{equation}

\T{(2) Router invariants.} 
Every router must also obey flow conservation, \ie the total incoming traffic should equal the total outgoing traffic:

\begin{equation}
\sum_{\forall l} ~l^X_{in} = \sum_{\forall l} ~l^X_{out} 
\label{eq:router}
\end{equation}

\T{(3) Path invariants.} 
The traffic demand estimated using the forwarding tables should match the load observed on the corresponding link. 

\begin{equation}
l_{demand} = l^X_{out} = l^Y_{in} 
\label{eq:path}
\end{equation}

\begin{figure*}
  \centering
    \begin{subtable}[b]{0.24\textwidth}
    \centering
    \resizebox{\linewidth}{!}{
      \begin{tabular}{lc}
        \toprule
        \textbf{Invariant Status} 
        & \textbf{Fraction of total time} \\
        \midrule 
        {Agreement} & 99.98\% \\
        {Disagreement} & 0.02\% \\
        \bottomrule
      \end{tabular}
      }
      \caption{Link status invariant}
      \label{fig:inv:d}
  \end{subtable}
  \hfill
  \begin{subfigure}[b]{0.24\textwidth}
    \centering
    \includegraphics[width=\textwidth]{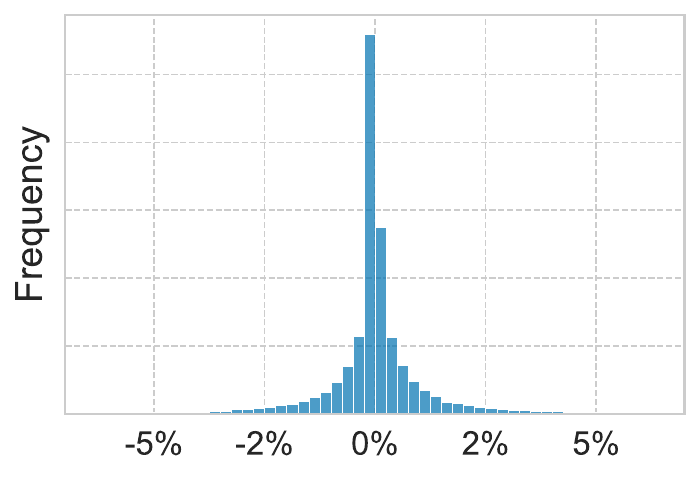}
    \caption{Link invariant}
    \label{fig:inv:a}
  \end{subfigure}
  \hfill
  \begin{subfigure}[b]{0.24\textwidth}
    \centering
    \includegraphics[width=\textwidth]{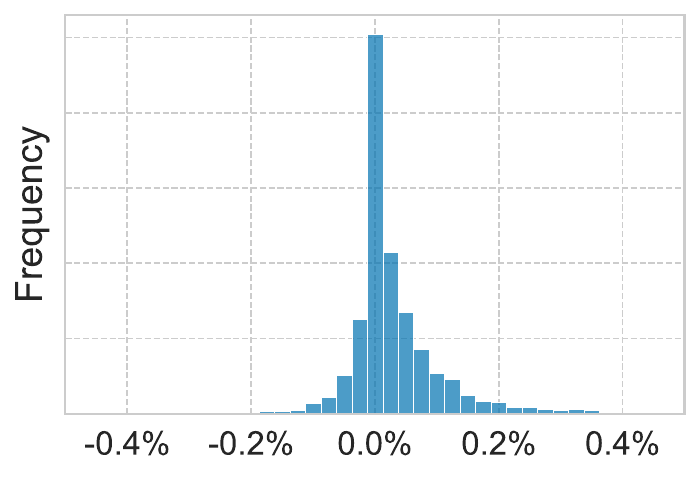}
    \caption{Router invariant}
    \label{fig:inv:b}
  \end{subfigure}  
  \hfill
  \begin{subfigure}[b]{0.24\textwidth}
    \centering
    \includegraphics[width=\textwidth]{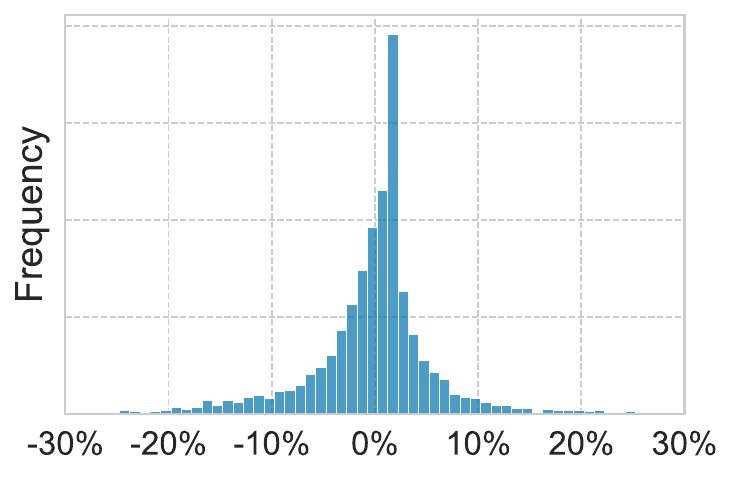}
    \caption{Path invariant}
    \label{fig:inv:c}
  \end{subfigure}
  \caption{Measured imbalance in our network invariants for a large production WAN. For (b)-(d), $0$\% implies that the equality holds perfectly.
  }
  \label{fig:inv}
\end{figure*}

\T{Why invariants may not hold.} While each of the above invariants is intuitive and is expected to hold in an ideal network, many regular real-world effects can conspire to break these invariants in practice, even without the occurrence of any bug. Invariants may not hold due to effects in (1)~the data plane, including packet propagation, queuing and processing delays, since these invariants assume zero delays, as well as packet drops, since the path invariant assumes that all demand arrives at all routers along its path; (2)~the control plane, since paths are not updated atomically across all routers; and (3)~the measurement layer, for example due to  measurement noise and loosely-synchronized clocks. 

\T{Measured real-world invariants.} To confirm \cameraready{the degree to which} these invariants hold in practice \cameraready{in the face of} the above real-world effects, we analyzed router telemetry data from two large production WANs $A$ and $B$ with O(100) nodes and O(1000) nodes respectively, over five-minute windows in a two-week period. \cref{fig:inv} depicts the results for WAN $A$. The results for WAN $B$ are similar and presented in \cref{sec:WAN_B}.

We first measured the percentage of link status reports in which our status-agreement invariant holds, \ie 
the routers on either side of a link agree on whether the link is up or not.  
We see disagreement on the link status only 0.02\% of the time and hence our invariant holds 99.98\% of the time (\cref{fig:inv:d}). 

We also check the degree to which our load-related invariants 
(\ie Eqs.~(\ref{eq:link-counters})--(\ref{eq:path}))
hold in practice. \cref{fig:inv:a} examines the link invariant derived from flow conservation: we plot the PDF of the difference between the link counters on each end of a link  (Eq.~\ref{eq:link-counters}) and see that they differ by less than 4\% for 95\% of links.  
Similarly, \cref{fig:inv:b} examines our router invariant, plotting the PDF of the difference between incoming \vs outgoing rates at a router. We see that our router invariant (Eq.~\ref{eq:router}) holds within 0.21\% for 95\% of routers. 
This is the tightest load invariant, because all measurements are local to a router.
Nonetheless, minor differences remain due to measurement  offsets, packets in flight through the router, and dropped packets. 
Finally, \Cref{fig:inv:c} examines our path invariant, plotting the PDF for the difference between $l_{demand}$ and the average of $l^X_{out}$ and $l^Y_{in}$ (Eq.~\ref{eq:path}). We see that the path invariant  holds fairly well, but with a larger tail: for 75\% of links the difference is within 5.6\%, with a difference of 15.3\% at the 95-percentile. 
The main reason is that paths may be updated during the measurement period and collecting path information is not synchronized (across routers or with path updates), introducing discrepancies between the load estimated from demand inputs \vs the measured link loads. 

Since our load-related invariants do not hold exactly, we define a threshold $\NN$ and say that an invariant holds if the relevant measurements are within $\NN$ of each other. In our experiments, we set $\NN=5$\,\%, which corresponds to the 96.5th, 100th and 71.7th percentiles in \cref{eq:link-counters,eq:router,eq:path} respectively. \sys will then use this threshold $\NN$ in the repair algorithm to determine when two estimates for a link load can be deemed as equivalent (\cref{sec:repair}).

\section{\sys Repair and Validation} 
\label{sec:design}

We now describe \sys's repair and validation steps in detail.
Recall that the two inputs to be validated are \textit{demand} and \textit{topology}.
Since both rely on a common repair step to construct a reliable view of the current network state, we begin by describing this shared repair process (\cref{sec:repair}).
We then present the validation logic, which is different for each input (\cref{sec:design-validating-demand} and \cref{sec:design-topology}).

\subsection{Repair}
\label{sec:repair}
The goal of \sys's repair algorithm is to derive a reliable value for the traffic load on each link (denoted $l_{final}$).
\sys then uses the derived $l_{final}$ not only to validate the demand (by comparing it against {$l_{demand}$}), but also as an additional source of ground truth for validating the topology, since $l_{final} > 0$ if and only if the link is up. 
For simplicity, we describe our repair algorithm informally in this section; the full algorithm is shown in \cref{s:dtp-pseudocode}.

\T{Voting.} The general strategy of our repair algorithm is to accumulate multiple estimates for a link's load, each obtained from different sources. We call these estimates \textit{votes}, and rely on a simple majority vote to select $l_{final}$.
Given a link $l$ from $X$ to $Y$, the collected router signals already give us three votes: $l_{demand}$, $l^X_{out}$, and $l^Y_{in}$.  Including $l_{demand}$ as a vote might seem non-intuitive given it is computed from the high-level demand inputs. However, this is a deliberate choice we make to avoid false positives in the face of incorrect router telemetry: \ie because $l_{demand}$ is independent of router counters, it can "vote against" the estimates derived from buggy counter values. Our evaluation in \S\ref{sec:eval-factor-analysis} confirms the value of this design choice. 
However, simply taking the majority of these 3 votes does not give us sufficient resilience to correlated bugs. 
For example, a counter bug at both routers $X$ and $Y$ could impact $l^X_{out}$ and $l^Y_{in}$ in the same way. 

\begin{figure}
\centering
\includegraphics[width=0.9\columnwidth]{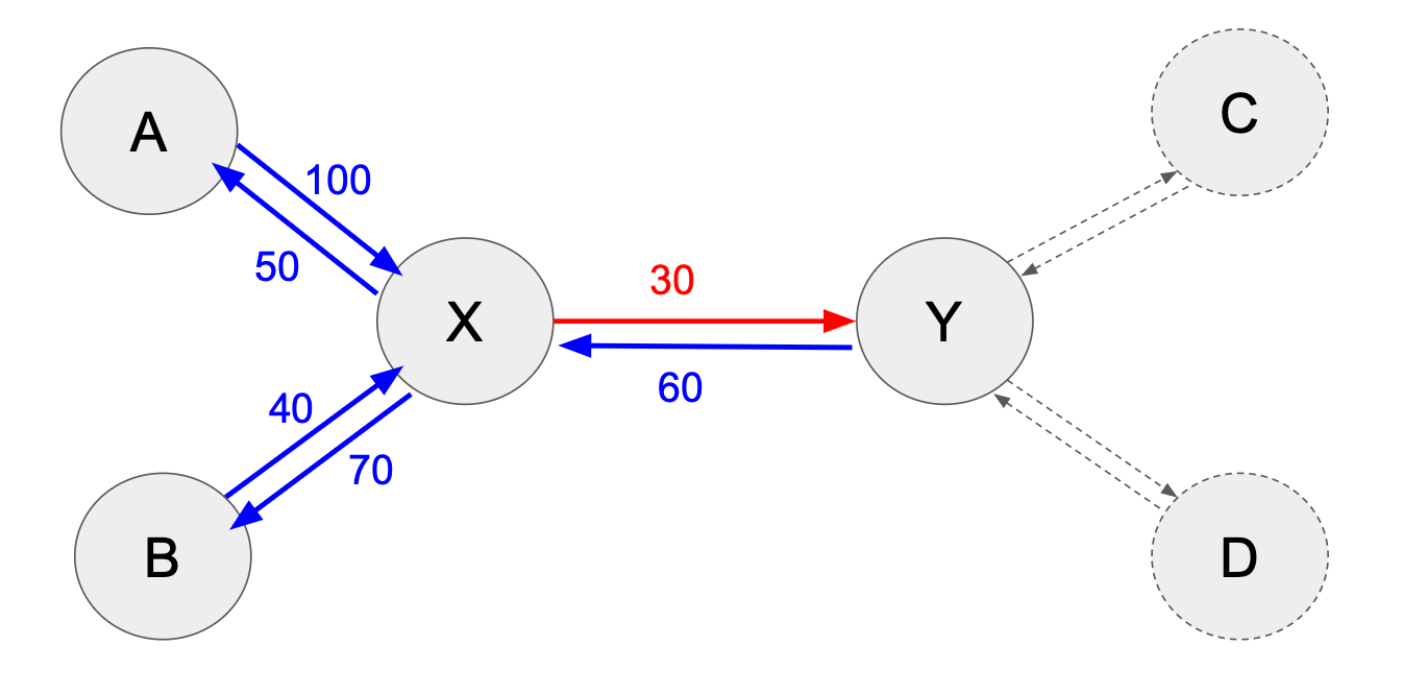}
\caption{Example network with a faulty router signal (in red).}
\label{fig:DTP-fig}
\end{figure}

\T{Deriving additional votes from router invariants.} 
\sys derives additional votes based on the observation that signals across routers are heavily correlated via the router invariants.
For example, consider the link $X \rightarrow Y$ in Fig.~\ref{fig:DTP-fig}.
Applying our router invariant at $X$ (\ie incoming load equal outgoing load) using the blue links, \cameraready{we can estimate the load on} $X\rightarrow Y$ as $(100+40+60)-(50+70)=80$. Doing the same at $Y$ would yield an additional estimate. Thus we can obtain 2 additional votes by taking into account the load on the \emph{neighboring} links of $X \rightarrow Y$. 

However, this \cameraready{raises} the question: When computing the router invariant at each router, which vote should \sys use for each link?
Specifically, while the example in Fig.~\ref{fig:DTP-fig} shows a single value for the load on each neighboring link, in practice, it can choose three baseline estimates for each neighbor link. 
For example, the load on the link $m$ from $A\rightarrow X$ could be estimated using $m^A_{out}$ or $m^X_{in}$ or $m_{demand}$.

Ideally, we would like to leverage all possible votes, since each of them comes from a distinct source and thus increases our tolerance to correlated bugs across multiple routers. 
For example, considering $m^A_{out}$ when we apply router invariants at $X$ could allow us to catch bugs that impact $X$ but not $A$. More generally, if a sufficient number of the routers neighboring $X$ and $Y$ are bug-free and we start counting their votes, we could obtain a majority vote that overrides the incorrect counter estimates provided by $X$ and $Y$ (as later formalized in \cref{thm:dtp}).

\T{Multiple voting rounds.} Taking into account all the votes for each link rapidly leads to a state explosion problem when we consider all possible combinations for each link and router. %
For instance, just in the small example of Fig.~\ref{fig:DTP-fig}, there are three estimates for each of the five blue links that surround router $X$, yielding $3^5=243$ tuples of estimates. 
To address this, the \sys repair algorithm runs many \emph{rounds} of voting. 
In each round, it randomly picks one of the three candidate values for each link $l$ ($l_{demand}$, $l^X_{out}$, and $l^Y_{in}$), and then applies the router invariant at each router in the resulting topology.
Thus, after $N$ rounds of voting, each router $X$ has $N$ votes for the load at each of its incident directed links. We then merge the votes in agreement, \ie those within the previously-defined threshold of $\NN=5\%$,  
taking their average as their representative. 
The value cluster with the maximum number of votes serves as $X$'s estimate for that link's load as derived from router invariants, which we denote as $l^X_{rtr}$. In addition, we compute a corresponding  weight score (denoted $w^x_{rtr}$) which is simply the fraction of the $N$ votes that were cast for $l^X_{rtr}$. 

\T{Consolidating votes: from five to one. }  
Our procedure so far results in five estimates for each link $l$'s load: 
$l^X_{out}$, $l^Y_{in}$, $l_{demand}$, $l^X_{rtr}$ and $l^Y_{rtr}$. 
We assign a weight of $1.0$ to the first three estimates, and weights $w^x_{rtr}$ and $w^y_{rtr}$ as above to the last two.
\footnote{We choose $1.0$ for the first three estimates based on the intuition that they rely on direct measures at $l$'s incident routers $X$ and $Y$; our evaluation shows that this simple choice performs sufficiently well (\S\ref{sec:eval}).} 
Our next step is to derive a single final estimate (\ie $l_{final}$) from this set. 
For this, we cluster votes that are within our noise threshold, summing the weights of any clustered values, and then pick the estimate with the largest cumulative weight as our \emph{tentative} final estimate for $l$'s load, with the corresponding cumulative weight as our confidence score for this estimate.  
However, before proceeding to finalize a link's load estimate, we undertake one final step as described below.

\T{Gossip before finalizing.} Our procedure so far uses a one-shot vote approach, which leaves us vulnerable to local pockets of correlated bugs. For example, in Fig.~\ref{fig:DTP-fig}, repair fails if routers $X$, $Y$, $A$, and $B$ all suffer from the same counter bug, since only the $l_{demand}$ estimates would be exempt from the bug, and these would be out-voted by the counter-derived estimates. 

Loosely inspired by gossip algorithms, we want to give a chance for values with high confidence to propagate and influence other values. We achieve this by using the repair algorithm \emph{iteratively}. At each iteration, \sys only finalizes the link with the highest confidence, 
and then repeats the above process of running multiple rounds, tallying votes, and so forth. Namely, \sys picks the link $\tilde{l}$ with the highest confidence score and assigns its (previously tentative) final estimate to $\tilde{l}_{final}$.
Crucially, once finalized, a link's estimate is fixed in all subsequent iterations. 
Thus for a topology with $L$ links, \sys runs the above repair process $L$ times, finalizing one link's estimate each time, until it arrives at a value of $l_{final}$ for each link $l$. 

At this point, \sys has a reliable (though not necessarily perfect) value for the load on each link. 
We next describe how it uses this value to validate the inputs to the SDN controller, beginning with the demand input, followed by the topology input.

\subsection{Validating Demand}
\label{sec:design-validating-demand}

When validating demand, \sys needs to answer the question:
Given the repaired link loads ($l_{\mathit{final}}$), how can we determine whether the demand matrix provided to the SDN controller is correct, in a manner that ensures a near-zero FPR and a high TPR?

A natural approach is to apply the updated path invariant, \ie check whether the demand-induced load $l_{\mathit{demand}}$ matches the repaired link load $l_{\mathit{final}}$, and alert the operator when the two differ. However, this raises the question: is any single violation of this invariant sufficient evidence of a faulty demand input? Or could such violations arise from the larger noise in path invariants and/or residual errors in router telemetry that were not fully corrected during the repair step?

To correctly identify incorrect inputs, \sys leverages the observation that \emph{the pattern of inconsistencies between $l_{\mathit{demand}}$ and $l_{\mathit{final}}$ can reveal the underlying cause}. 
Specifically, incorrect demand inputs tend to induce \emph{widespread} violations of the path invariant: for example, a single erroneous entry $D_{IJ}$ in the demand matrix will manifest as inconsistencies along all links traversed by traffic between routers $I$ and $J$. In contrast, noise and residual faults in the telemetry of one router produce \emph{localized} inconsistencies, typically confined to the set of its adjacent links.

To check for widespread violations of the path invariant, \sys proceeds in two steps. First, given some threshold $\tau$, it considers that a path invariant holds at a link $l$ when its imbalance ($l_{\mathit{final}} - l_{\mathit{demand}}$) falls within $\tau$. 
\cameraready{Then, in order to make a top-level validation decision, \sys relies on a validation cutoff $\cutoff$, classifying} a demand input as correct when the fraction of links for which the path invariant holds is above $\cutoff$, \ie when it does not encounter widespread violations. 
On the other hand, a buggy demand will encounter few links where its imbalance falls within $\tau$, and therefore its fraction of satisfied links is expected to be below $\cutoff$, \ie it will be classified as incorrect.

At each new WAN, \sys sets $\tau$ and $\cutoff$ after an initial calibration phase, where it collects telemetry data and input demand matrices during a known-good period (\eg after 
an operator confirmation of stability). Given the collected path imbalance distribution during this calibration phase, $\tau$ is automatically set at the 75th percentile of this distribution.\footnote{This is a heuristic: a large percentile would accept large imbalances and therefore miss bugs that affect small demand volumes, while a small percentile may be too sensitive to counter perturbations.}
Then, for each recorded time interval, \sys applies the repair procedure, computes the number of links satisfying the path invariant, and records the resulting fraction. 
To maintain a near-zero FPR, \sys sets $\cutoff$ to \cameraready{just below} the minimum value observed across this calibration window. For example, in WAN $A$, \sysname first went through a two-week calibration period. It set $\tau=5.588$\,\% as the 75th percentile of path-invariant imbalance, and $\cutoff=71.4$\,\%.

At runtime, \sys validates demand inputs by comparing the observed consistency rate to this threshold. If the fraction of links on which the path invariant holds exceeds $\cutoff$, the input is classified as \emph{correct}; otherwise, it is flagged as \emph{incorrect}.
\cref{fig:psuedo-dtp-validation} summarizes this procedure in pseudocode. 

\T{Configuring hyperparameters.} 
\cameraready{
CrossCheck's validation algorithm relies on four key hyperparameters mentioned above that balance detection sensitivity against false positive rate. We summarize and compare them here:

\Ts{(1) The noise threshold} $\NN$ accommodates inherent measurement noise in production networks by defining when two load estimates are considered equivalent during the repair process. The network operator sets it based on an empirical analysis of the distribution tail of the difference in counters, which we set conservatively to be 5\% given the link and router distributions shown in \cref{fig:inv:b,fig:inv:c}. 

\Ts{(2) The number $N$ of voting rounds} in the repair algorithm  determines how many random combinations of link estimates are explored when applying router invariants, with more rounds providing greater resilience to correlated failures at the cost of increased computational overhead. We found $N=20$ to be effective for our network. The optimal $N$ is correlated with average node degree.

\Ts{(3) The imbalance threshold} $\tau$ determines the acceptable discrepancy between demand-induced load estimates and repaired link loads for individual links, effectively controlling the amount of noise in top-level invariants that our system is tolerant of. It is set by observing the distribution of invariant imbalances over a system observation period, which we show for our network in \cref{fig:inv:c}. 

\Ts{(4) The validation cutoff} $\cutoff$ applies to the final top-level validation decision, specifying what fraction of interfaces must satisfy the path invariant for the entire demand input to be classified as correct. This threshold is critical for maintaining near-zero false positive rates and is likewise calibrated during an initial observation period by setting it to be just below the minimum consistency rate observed during known-good operation.}

{
    \SetAlgoNoEnd
    \begin{algorithm}[t]
    \SetAlgoLined
    
    satisfied\_count = 0\;
    \ForEach{link l}{
        \If{percent\_diff(l.demand, l.final) $\leq \tau$}{
          satisfied\_count += 1\;
        }
    }
    
    \Return satisfied\_count / num(links) > $\cutoff$\;
    \caption{Demand validation.}
    \label{fig:psuedo-dtp-validation}
    \end{algorithm}
}

\subsection{Validating Topology}
\label{sec:design-topology}

Validating the topology input is comparatively straightforward, as \sys can rely on \emph{five} independent signals to directly corroborate the status of each link. 
These signals are considered independent because they originate from distinct subsystems within the router.

For a directed link $l$ from router $X$ to router $Y$, \sys uses the following signals: $l^X_{\mathit{phy}}$, $l^Y_{\mathit{phy}}$, $l^X_{\mathit{link}}$, $l^Y_{\mathit{link}}$, and whether $l_{\mathit{final}} > 0$.
The first two represent the physical-layer status as reported by $X$ and $Y$, respectively.
The next two reflect the link-layer status as determined by heartbeat protocols like BFD.
The fifth signal, $l_{\mathit{final}}$, reflects observed traffic on the link after the repair step and is computed using counters from across the network, making it independent of the others.

Given these independent sources, \sys applies a simple majority vote across the five signals to determine the operational status of each link.
In our evaluation, this approach successfully resolved the rare 0.02\% of cases where the link status invariant was violated (\cref{fig:inv:d}), and detected all instances of incorrect topology inputs encountered in our log of outages.

\subsection{Analyzing {\sys}’s Algorithms}

\T{Repair guarantees.}
Assume that the input demands and paths are correct, and that we set the threshold $\NN$ high enough to capture regular noise when consolidating votes in the repair algorithm. 
Then the \sys repair algorithm can provably detect and repair any corrupted counters at an arbitrary single (directed) link if counters at other links are fine. \cameraready{We show the full proof in \cref{sec:proof}.}

\begin{theorem}\label{thm:dtp} 
\sys is guaranteed to detect and repair any corrupted counters when corruption is restricted to an arbitrary single link.
\end{theorem}
For example, in \cref{fig:DTP-fig}, even if the egress counter of $X$ and the ingress counter of $Y$ at link $X \to Y$ are \textit{both} corrupted (and the rest of the network only suffers from normal noise), \sys is guaranteed to repair link $X \to Y$. 

\medskip
\T{Scalability of demand validation with network size.}
Evaluating the consistency between $l_{\mathit{demand}}$ and $l_{\mathit{final}}$ at a network-wide level ensures that \sys becomes more accurate as the network grows; \ie the FPR decreases, and the TPR increases with larger networks. 

Formally, if we assume that the path-invariant imbalance follows an i.i.d. distribution across all links of all networks when given some healthy (resp., faulty) input demands, then:

\begin{theorem}\label{thm:exp}
As the number of links $n$ increases, both FPR and ($1 - \text{TPR}$) converge to zero exponentially fast.
\end{theorem}

The full proof is presented in \cref{sec:proof_exp}; we summarize the key intuition here.

\textit{Why the FPR decreases with network size.}
As the number of links grows, the distribution of invariant deviations observed across the network more closely matches the true underlying distribution. This happens because, with more samples (i.e., links), the overall pattern becomes more reliable. The DKWM inequality~\cite{massart1990tight,reeve2024short} 
formalizes this intuition by bounding how far the observed distribution can deviate from the true one, and it shows that such deviations become exponentially smaller as the number of links increases. Thus, in large networks, the fraction of links satisfying the invariant stabilizes, allowing \sys to set conservative thresholds that minimize false positives with high confidence.

\textit{Why TPR increases with network size.}
Larger networks naturally involve more traffic demand entries and more paths per demand. Hence, even a small number of incorrect demand entries affect a large portion of the network state.

Consider a 150-router network with 100 border routers and 50 transit routers. This setup yields $100^2 = 10{,}000$ distinct demand entries. If each demand traverses 5 routers on average and is multipath-routed over 4 disjoint paths, then a single demand results in roughly $5 \times 2 \times 4 = 40$ counter readings (ingress and egress at each router). Across all demands, this amounts to $10{,}000 \times 40 = 400{,}000$ total signal contributions.

Now, assume 5\% of demand entries are incorrect. These affect $0.05 \cdot 400{,}000 = 20{,}000$ signal contributions. Spread over an estimated $150 \cdot 10 = 1{,}500$ router interfaces (assuming average degree 5 and 2 counters per port), each interface experiences over $13$ error-induced discrepancies on average.

Thus, this analysis shows that even a modest fraction of incorrect demand inputs creates a dense pattern of invariant violations, allowing \sys to detect such faults with high confidence in large networks.
\section{Implementation} 
\label{sec:implementation}

We built a prototype implementation of \sys that performs input validation, using the architecture shown in \cref{fig:system-design}. 
The \sys prototype consists of two main components. The lower half is network-specific and handles the collection and storage of router signals, while the upper half is network-agnostic and performs repair and validation over these signals \cameraready{when provided as} input. 
\sys has been deployed as an experimental shadow system at a large cloud provider WAN (WAN~$A$), where it has been continuously validating the inputs to their TE controller over a four-week period.

\T{Collecting and storing router signals.}
The lower half of \sys is responsible for collecting and storing the raw telemetry required for validation. 
This component is intentionally kept simple and does not perform any of the aggregation performed by the control infrastructure, to reduce the chances of bugs.
Its implementation is specific to the telemetry interfaces available in the target network.

In WAN~$A$, \sys collects all telemetry via gNMI~\cite{gnmi, openconfig}, a standardized telemetry API supported across major router operating systems. 
\sys subscribes to physical and link-layer status event updates for each link, and samples byte counters every 10 seconds per interface~\cite{openconfig-interfaces}, \cameraready{emitted as} a stream of (\textit{timestamp}, \textit{total-bytes-in/out}) tuples from which transmit and receive rates are derived. 
Forwarding paths are reconstructed using abstract forwarding entries retrieved from routers, which specify how traffic is split and forwarded across tunnels. All telemetry is stored in an in-house, in-memory time-series database that is similar to systems such as TimescaleDB~\cite{timescale}, InfluxDB~\cite{influx}, and Monarch~\cite{monarch}.

A natural question is whether such a ``flat'' system that performs no aggregation can  scale. To answer this we calculate the write rate for a moderately-large network, storing roughly 10 metrics every 10 seconds from $O(10,000)$ interfaces. This rate of $O(10,000)$ writes per second is well within the peak performance of up to 2.4 million inserts per second shown for open-source time-series databases in prior work~\cite{timescale-white-paper}, and is easily accommodated by our in-house database system.

Finally, since raw router counters report monotonically increasing byte totals timestamped at each sample, \sys must convert these into byte-per-second rates for use in repair and validation. 
Our timeseries database supports a query language, which \sys uses to issue a short query~--~just five lines~--~that aggregates interface counters into bundles and computes rate estimates over time. 
Occasionally, counters may reset due to hardware overflows or router restarts; \sys explicitly detects and excludes such intervals from rate computation to avoid introducing spurious artifacts.

\T{Repair and validation.}
The upper half of \sys is responsible for consuming telemetry and executing the repair and validation procedures described in \S\ref{sec:design}. This component is designed to be general-purpose: it interfaces with any backend time-series database via a pluggable API that abstracts over telemetry retrieval.

The implementation is lightweight. The core repair and validation logic is written in Python and comprises approximately 250 and 100 lines of code, respectively, supplemented by a few hundred lines of orchestration code for scheduling runs and logging outputs. 
\sys exposes a simple \texttt{validate(demand, \cameraready{topology})} API, which takes a demand matrix and topology as input and returns a binary validation decision.

\T{Deployment as a shadow system.}
The WAN is mission-critical in cloud infrastructure and hence we are unable to directly integrate our research prototype into the critical path of the production TE system. However, we are able to run \sys as a ``shadow'' system (that is off the critical path) by leveraging an independent storage replica of the live TE database used in WAN~$A$. Thus \sysname runs on real-time inputs (demand and topology) and telemetry from WAN~$A$ but our validation \cameraready{decision} output is not integrated into the production TE system.  
\section{Evaluation}
\label{sec:eval}

We evaluate \sysname through a combination of deployment as a shadow system in a WAN (\cref{sec:eval-shadow-system}) and simulation (\cref{sec:eval-simulation}). The former allows us to test \sysname under realistic conditions, while the latter allows us to explore a range of possible error scenarios we did not encounter in production. Our metrics are the TPR and FPR, as defined in ~\cref{sec:introduction}, and our goal is to ensure a near-zero FPR, while maximizing the TPR. 

\subsection{Shadow Deployment Evaluation}
\label{sec:eval-shadow-system}

Here, we seek to answer the following questions: 
(1)~Will the noise inherent to production data trigger false positives? (2)~Will \sys be able to detect incorrect inputs in the real world? and (3)~What are \sys's performance characteristics when run on production data volumes? 

To answer these questions, we deployed \sys as an experimental shadow system in a cloud WAN (WAN~$A$), where it has been continuously validating the inputs to their TE controller over a four-week period. 
In porting \sysname from the lab to production, we discovered a handful of practical changes that were necessary. First, we adjusted our equality accounting for headers as the load estimates derived from demands were systematically 2\% lower than interface counters. We discovered that this is  because, depending on the vendor, counters on some routers include bytes due to packet headers while our demand inputs do not. Second, we observed that for datacenter-facing interfaces, we needed to include \textit{hairpinned} traffic---\ie traffic not routed on the WAN, but coming up and back down from the datacenter---when computing $l_{demand}$; again this traffic is not accounted for in the demand inputs but appears in router counters. 

\T{FPR and TPR in production.}
Over the 4 weeks that \sys was deployed, we saw \textbf{zero false positives}, showing that our \sys system is resilient to the inherent noise in real-world demand and telemetry data.

Additionally, our system  \textbf{correctly detected a rare instance of incorrect input data}: during a short period, the database holding demand was affected by a bug introduced in a new code release, which caused it to double-count the demand measured at the end hosts. As a result, all demands in this replica were doubled for most of a period of 3 days before the issue was detected manually and rolled back. 

\cref{fig:shadow-system-results} demonstrates the output of \sys's validation during the period that the input data was incorrect. 
We can see that incorrect inputs cause a steep drop in validation scores when demands are found to be inconsistent with the state of the network derived from the router signals.

\begin{figure}
\centering
\includegraphics[width=1\linewidth]{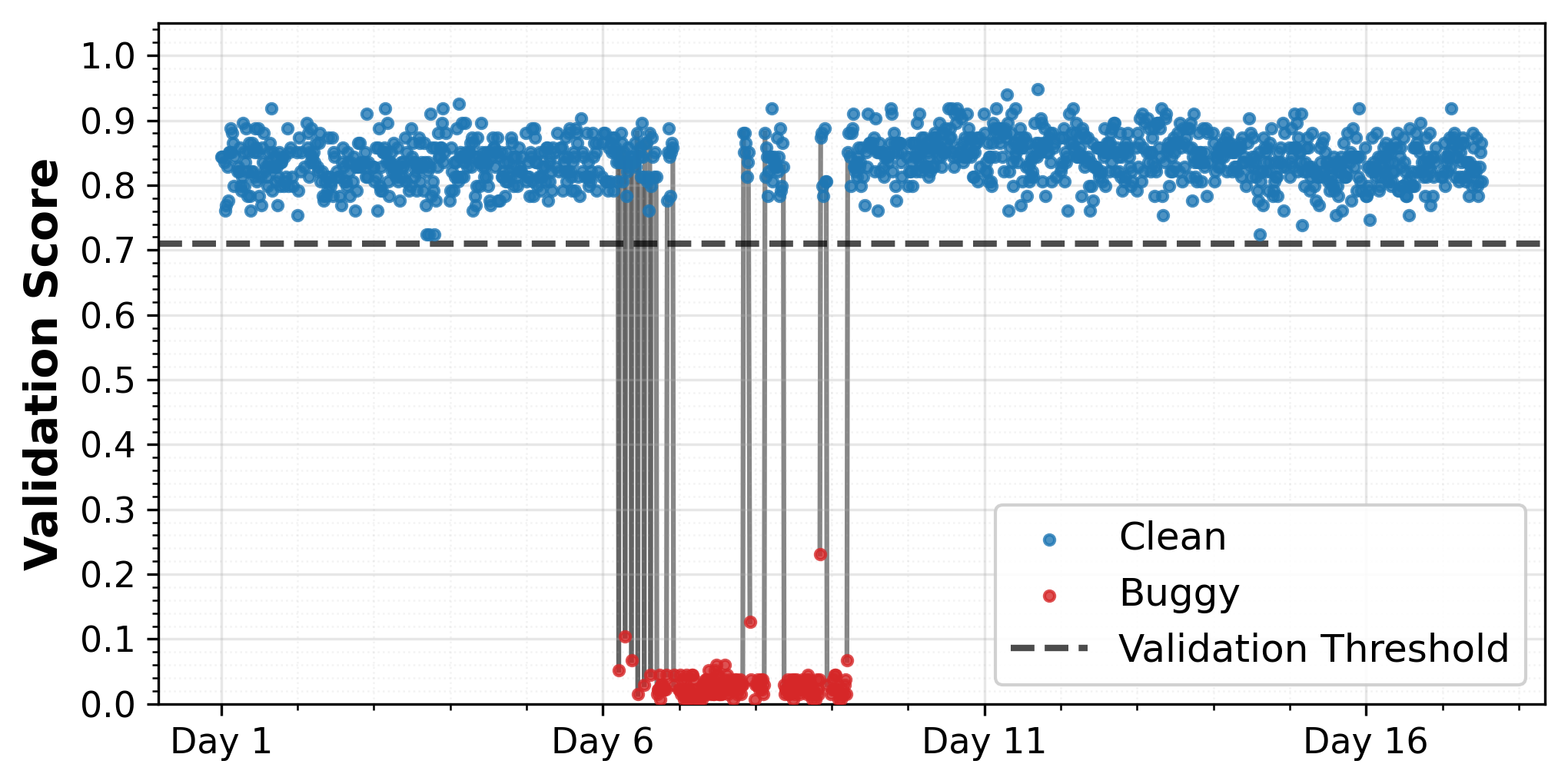}
\caption{Shadow-system validation on live production data from a large production WAN $A$.}
\label{fig:shadow-system-results}
\end{figure}

We highlight two takeaways from this incident. First, \sysname was able to detect an incorrect input where existing static checks failed to do so. Second, although the incorrect inputs did not affect our production TE controller (since this was a separate storage replica), they did affect some lower-importance capacity planning systems.\footnote{Capacity planning is done over a longer time-window \cameraready{and less frequently}, which is how the bug went undetected for multiple days.} Moreover, this class of bug might equally have impacted the production DB, leading to a real outage. Indeed, similar bugs triggered some of the outages we reported in \cref{sec:background}.

We also tested \sysname retrospectively over a separate storage replica that stores topology health data which is consumed by a network health sentry service that monitors and drains unhealthy links. Reading the production topology inputs and raw telemetry data, \sysname was able to successfully catch a scenario of invalid topology input that caused all healthy links at a router to be drained, leading to congestion causing an outage. If fully deployed, \sysname could have detected and prevented this outage, as it was able to detect that the links were up and healthy through the topology repair and validation process described in \cref{sec:design-topology}.

\T{System performance.}
Since \sys  must continuously validate inputs in real time, at the timescales at which SDN control decisions are made (typically minutes~\cite{swan, dsdn}), its runtime is equally important to its utility. 
We found that \sysname's runtime is within 10\,seconds on our large WAN inputs.
Router signals are typically available within the database within O(1\,s) of when they are produced by routers. 
We benchmarked the queries that produce aggregated and average counter values, and observed that they take $\sim$56\,ms to run.
The largest component of runtime is repair, which takes $\sim$9.1\,s to run, while validation takes O(100\,ms). Put together, our system runs well within our target window and adds minimal latency overhead to typical TE iteration times. For latency-sensitive use-cases, to prevent delaying the control system, \sysname can be run on inputs in \textit{parallel} to the control system beginning \cameraready{any computation necessary for decision making}, with the validation result checked before it finishes, allowing the control system to \cameraready{proceed with} any live action. \cameraready{We believe a more optimized repair algorithm implementation may be able to further speed up runtime.}




\subsection{Simulation-Based Evaluation}
\label{sec:eval-simulation}

We test our system’s performance across a wider range of potential operating conditions in simulation. We use production traces of demands, paths, and topology from WAN $A$ with $O(100)$ routers and $O(1000)$ links\footnote{We report uni-directional links, including those for network ingress/egress.}, as well as the open-source demands and topology from the Abilene (12 routers, 54 links) and \geant datasets (22 routers, 116 links)~\cite{sndlib, topohub}.\footnote{We choose these two topologies as they are, to our knowledge, the largest public topologies that release production demand measurements.} 
We assume all-pairs shortest-path routing for Abilene and \geant.

\T{Simulated telemetry.} Once we have the demand, topology, and path information, we can use the path invariants~--~\ie derive per-link loads by tracing the demands along their paths~--~to calculate the counter values for all interfaces in the network. However, these counter values are idealized since they ignore the noise we observe in practice. We thus \textit{add noise} to these idealized counter values. The amount of noise we add is carefully selected to match the real-world link-invariant, node-invariant and path-invariant noise distributions we reported in \cref{fig:inv}. We refer the reader to \cref{sec:gen} for the precise methodology we follow. The resultant counter values are our baseline for a \textit{healthy} network that has noise but no telemetry \cameraready{or demand} bugs.

\T{Dataset sizes.} Our snapshots of WAN $A$ are taken every 15 minutes over a four-week period in Spring 2025, totaling 2,000 snapshots, and our \geant and Abilene datasets comprise the first 4,000 snapshots of each. 

\T{Modeling buggy demands/telemetry.} We model bugs by applying a range of random perturbations. When fuzzing demand that is taken as input to TE, we first pick a random number between 5\% and 45\% of demand entries to perturb, and then pick a range at random from {5\%-15\%, 15\%-25\%, 25\%-35\%, 35\%-45\%} from which to sample the amount of traffic to change in each affected demand entry. We stop at 45\% as larger changes are easier to detect. At each entry, demand is either always removed to model bugs that omit demand, or removed and added with equal probabilities to model bugs causing stale demands. To plot, we then compute the total percent of absolute demand changed in the snapshot. For telemetry, we use the same method, and specify the percent of counters and the perturbation range used. 

\vspace{-0.5em}
\subsubsection{Results}
\vspace{-0.5em}

\begin{figure}
    \centering
    \begin{subfigure}[t]{0.48\columnwidth}
        \centering
        \includegraphics[width=\linewidth]{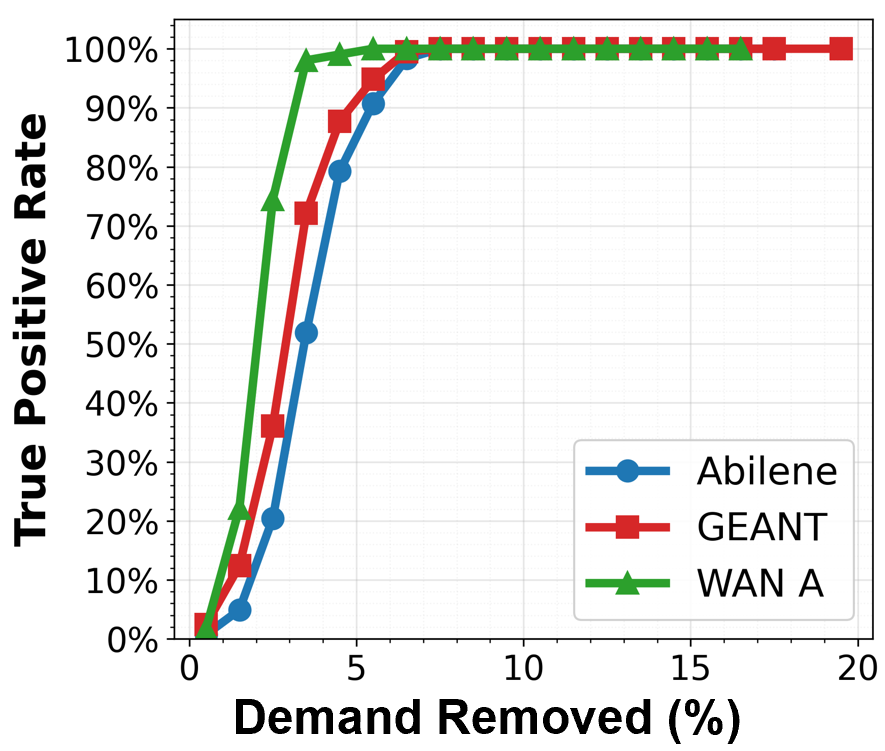}
        \caption{Demand removals.}
        \label{fig:eval-demand-sweep}
    \end{subfigure}
    \hfill
    \begin{subfigure}[t]{0.48\columnwidth}
        \centering
        \includegraphics[width=\linewidth]{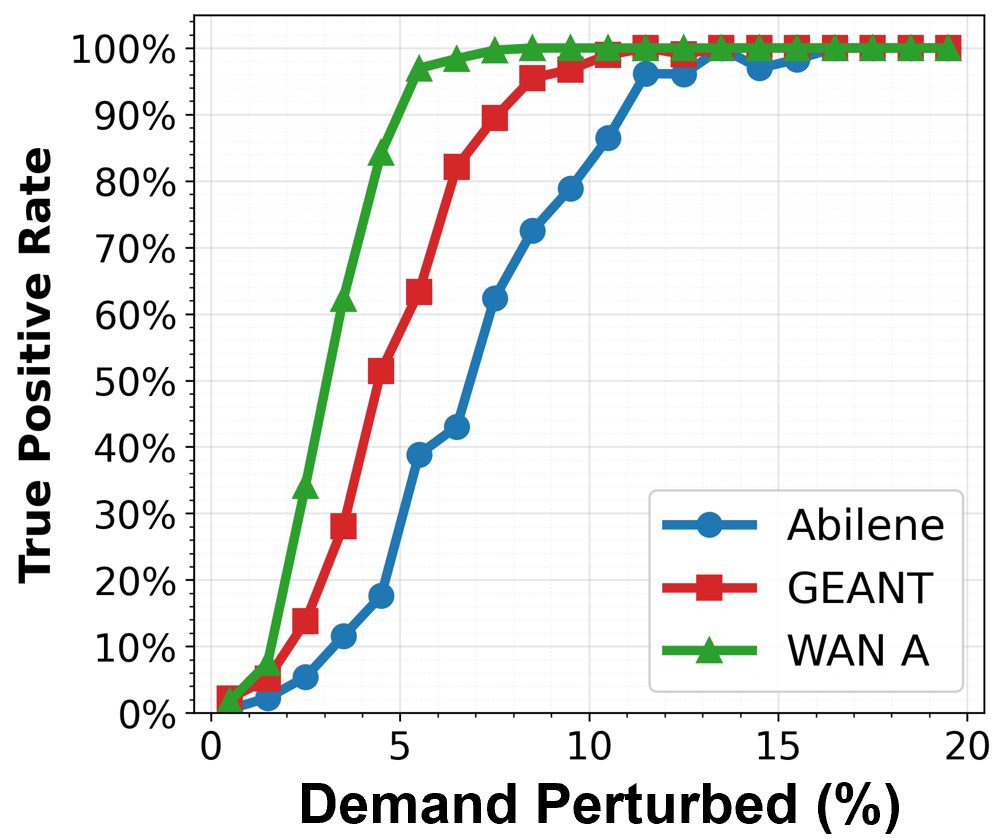}
        \caption{Demand removals and additions.}
        \label{fig:demand-random-sweep}
    \end{subfigure}
    \caption{\sysname's TPR with buggy demands. The x-axis displays the sum of the absolute values of the demand changes as a percentage of the total demand. 
    }
    \vspace{-5pt}
    \label{fig:demand-sweeps}
\end{figure}

\T{Is \sysname effective at catching perturbations to demand?} We sweep a range of demand perturbations and plot the resulting TPR as a function of the absolute change in demand. 
\cref{fig:demand-sweeps} depicts the results.
\Cref{fig:eval-demand-sweep} shows our results when modeling bugs that only removes demands. We see that for such errors, \cameraready{\textit{\sysname is very effective at detecting demand perturbations.}} It can detect 74\% of the 2-3\% perturbations, and 100\% of the 5\%+ perturbations. 

\Cref{fig:demand-random-sweep} shows our results when modeling bugs with stale demands (\ie in which demand elements are scaled up or down with equal probability). The outages we have observed only involved demands scaled in one direction, but we include stale ones as a harder test case.
We note that this is a particularly challenging scenario, as we are effectively shifting demand between flows, keeping the total demand constant on average. The results are only slightly worse for WAN $A$ compared to only removing demand  and therefore \sys does more than just checking the total demand. We see that very small networks (Abilene) are affected more greatly, as there is less path diversity, so additions to one demand cancel out subtractions from another when sharing the same links. Nonetheless, TPR remains high for larger changes in demand volumes, \eg approaching 90\% when 10\% of demands are perturbed.

As detailed in \cref{sec:design-validating-demand}, these strong TPR results are due to the fact that demand perturbations are reflected in many invariant checks (at each interface that demand traverses). Thus, a few demand changes can push enough validation invariant checks outside of the equality threshold, causing \sysname to detect the perturbed demand. We observe empirically that sensitivity increases with network size, as expected from \cref{thm:exp}. 

\begin{figure}
    \centering
    \begin{subfigure}[t]{0.48\columnwidth}
        \centering
        \includegraphics[width=\linewidth]{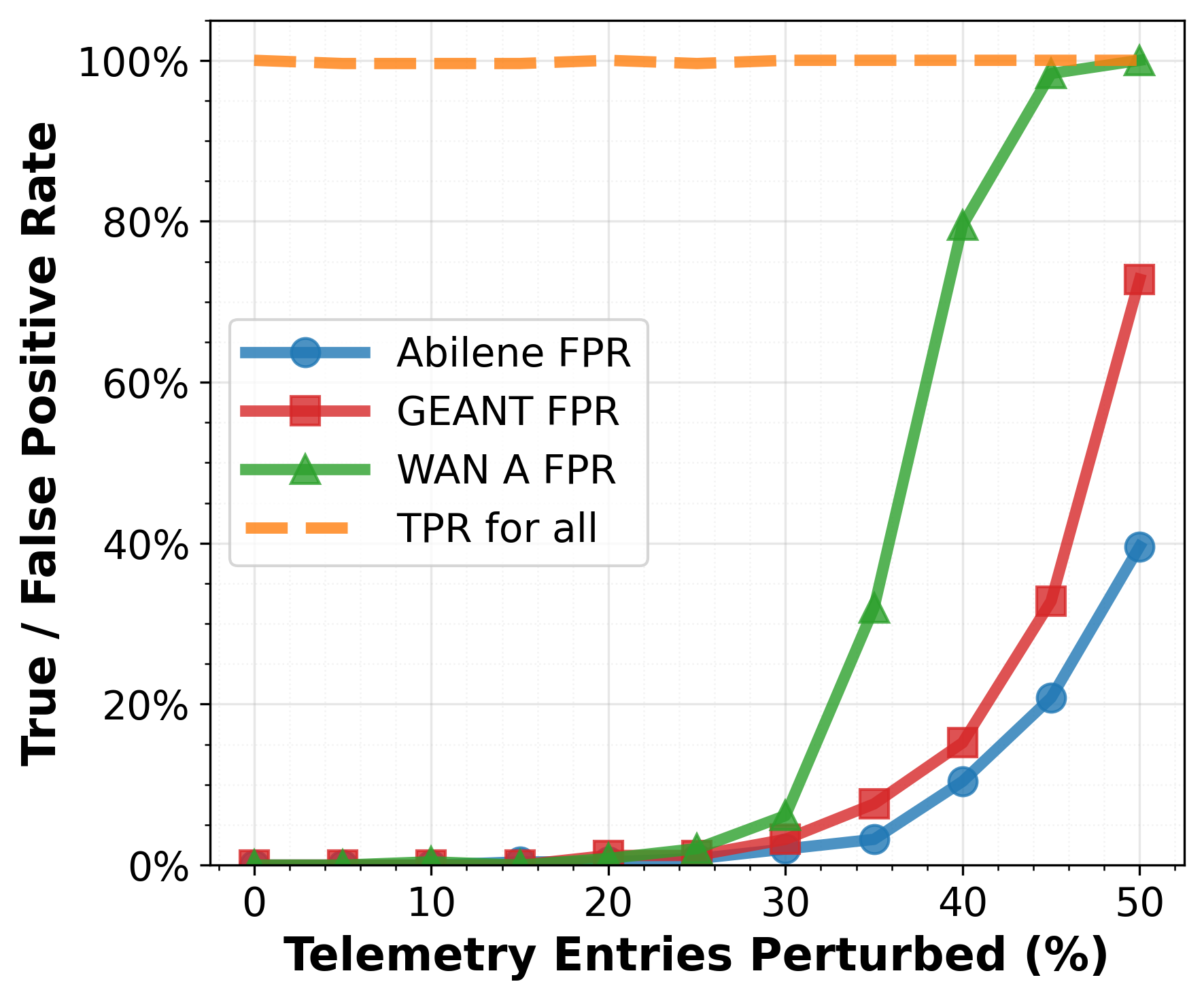}
        \caption{Random counter zeroing.}
        \label{fig:eval-telem-sweep}
    \end{subfigure}
   \hfill
    \begin{subfigure}[t]{0.48\columnwidth}
        \centering
        \includegraphics[width=\linewidth]{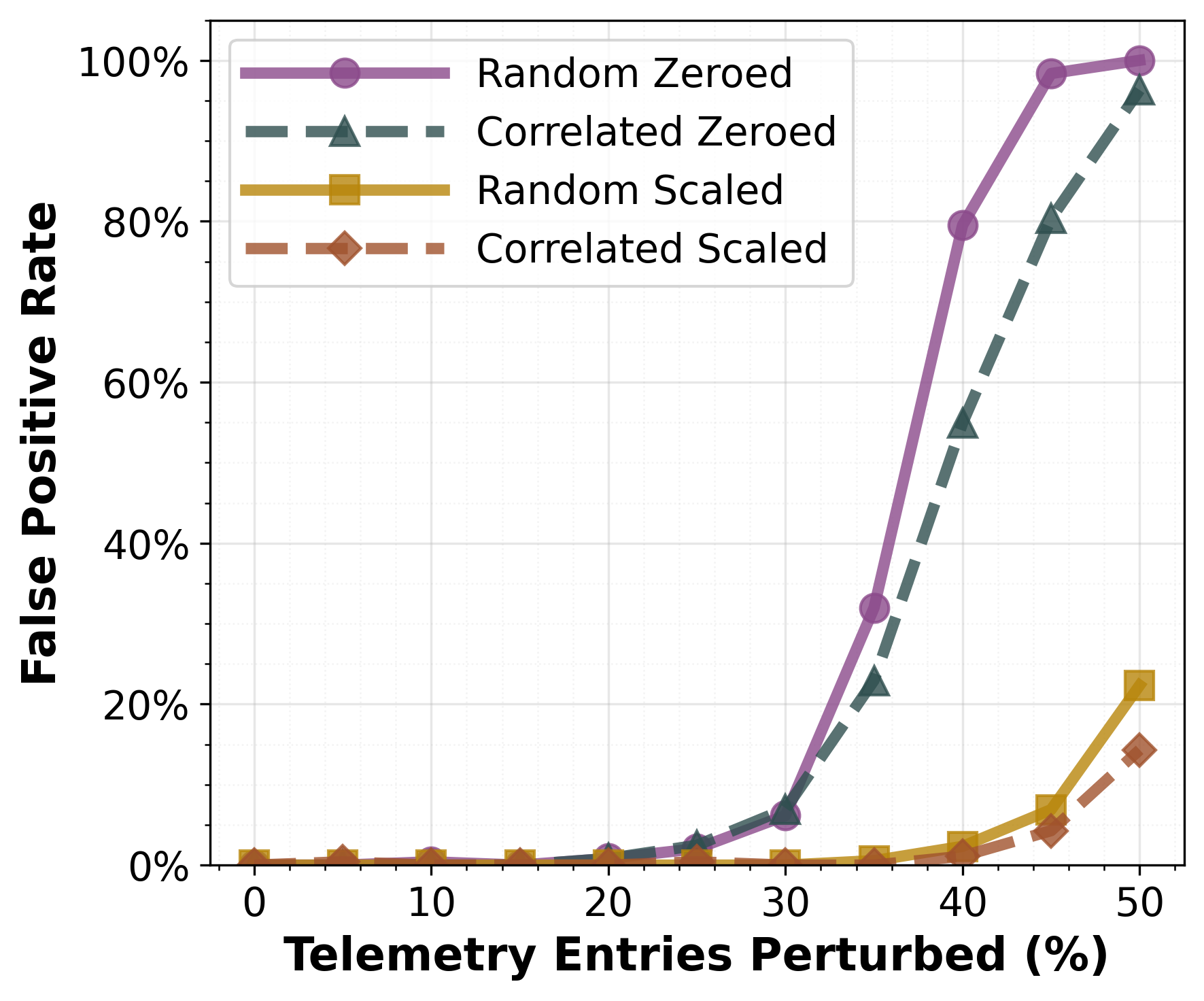}
        \caption{Four telemetry perturbations applied to WAN $A$.}
        \label{fig:eval-correlation-sweep}
    \end{subfigure}
    \caption{\sysname's FPR with buggy  telemetry.}
    \label{fig:topo-sweeps}
\end{figure}

\begin{figure}
    \centering
    \includegraphics[width=0.5\columnwidth]{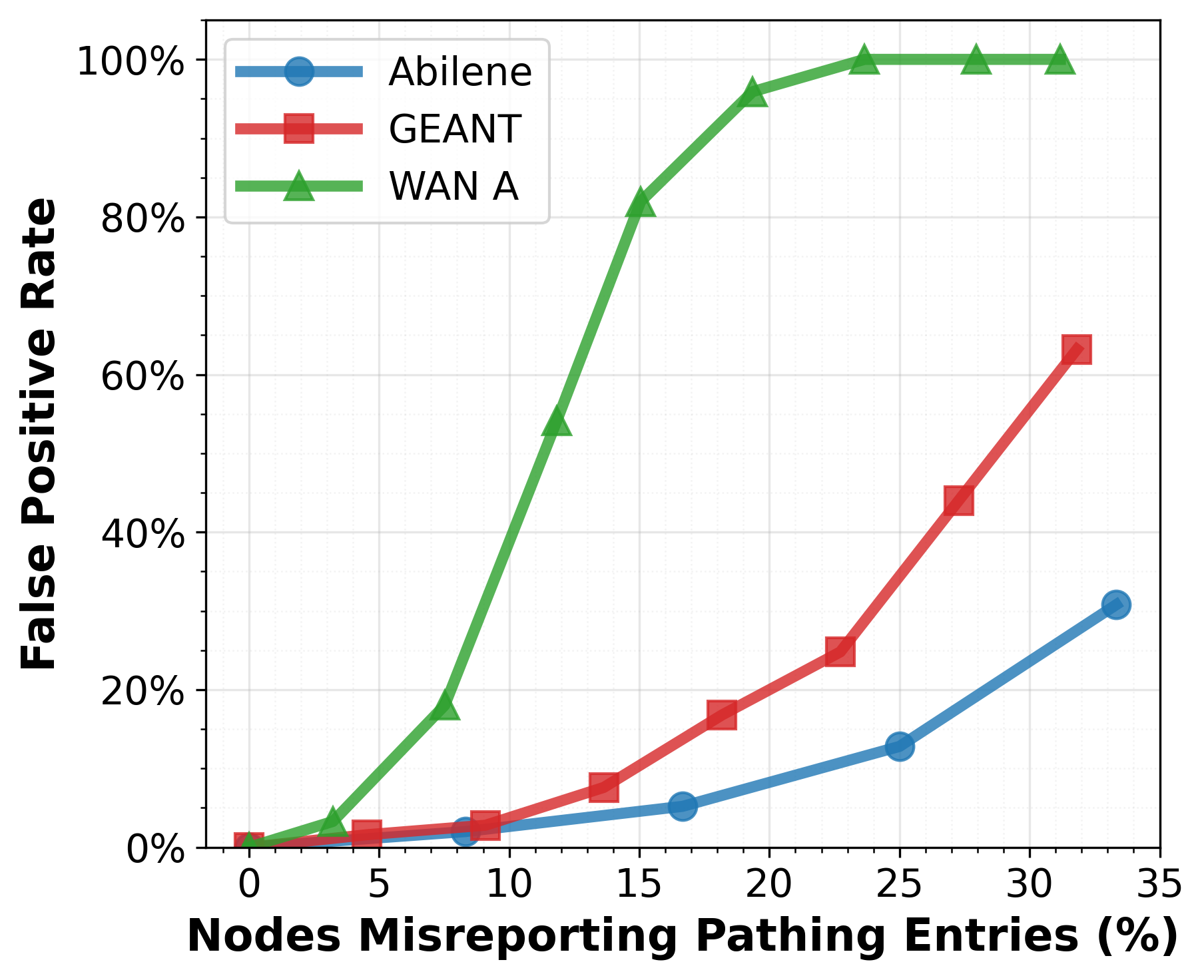}
    \caption{\sysname's FPR with buggy path information.}
    \label{fig:path-sweeps}
\end{figure}

\T{Is \sysname resilient to buggy counter telemetry?} Buggy telemetry can cause validation invariants to fail even though the demand is correct, leading to false positives. \sysname's repair process is designed to catch and rectify telemetry bugs, enabling demand validation to proceed safely. \cref{fig:eval-telem-sweep} shows the performance as we perturb telemetry by zeroing an increasing percentage of counter values, simulating dropped or missing telemetry, which is the most common form of telemetry corruption. We choose zeroing as such errors are harder to repair than random scaling, because if both sides of a link are zeroed, their agreement is difficult to overcome to recover the original value.

Our results show that \cameraready{\textit{\sysname is very resilient to buggy counter telemetry}}, able to withstand up to 30\% of telemetry counters being zeroed before experiencing false positives. We observe the (pleasing) trend that with larger topologies, \cameraready{\sysname becomes increasingly resilient} to telemetry perturbations. We also show in \cref{fig:eval-telem-sweep} that \sysname maintains 100\% TPR in the face of telemetry perturbations in \textit{all} cases (orange line, assuming 10\% of demand volume randomly removed). This is because the set of buggy demand inputs that make up the TPR already causes enough invariants to fail to be detected; perturbing telemetry randomly only causes \textit{more} invariants to fail.

\T{Do correlated failures challenge \sysname's effectiveness?} Telemetry perturbations often occur due to router-level bugs, and thus may be correlated, meaning they affect all of the interfaces of a particular router at once. 
\cref{fig:eval-correlation-sweep} shows four classes of telemetry failures, comparing correlated and fully random zeroing and scaling (by 25\%-75\%). We see that \cameraready{\textit{\sysname's repair process fully recovers from random and correlated failures}} with up to 25\% of the telemetry having errors, while the FPR starts increasing for higher percentages. 

Note that correlated failures do not seem to significantly increase FPR. This is because they do not increase the probability that both sides of a link agree about an erroneous link value, \eg with zeroing, or that a single side is wrong. Moreover, the repair algorithm was designed to poll other sources of information for the node-invariant vote, and thus to address correlated failures.

\T{Is \sysname resilient to buggy paths?}
The other telemetry \sysname relies on is forwarding entries that define paths.
Though very rare, a router can possibly fail to correctly report some or all of its forwarding entries due to either a hardware or software fault on the router. We evaluate a particularly pessimistic node failure mode where each affected router reports not having \textit{any} forwarding entries, sweeping the percent of routers in \cref{fig:path-sweeps}. Our results show that the FPR \cameraready{stays at zero up until} we go above 4\% of nodes experiencing faults. Such failures typically only affect one router, hence \cameraready{in practice} are \cameraready{typically well} below the \cameraready{point at which \sysname begins to experience false positives}. Further, such bugs are easily detected, and in such cases the best strategy would be to skip validation.

\subsection{Factor analysis}
\label{sec:eval-factor-analysis}

We now analyze whether \sys's repair step~--~and each of its components~--~contributes to ensuring its low FPR. 

\begin{figure*}
    \centering
    \includegraphics[width=1\textwidth]{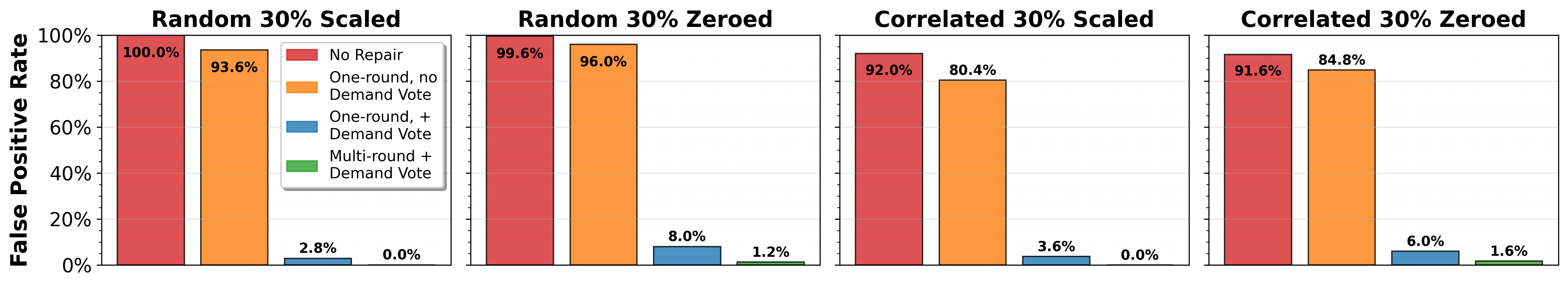}
    \caption{Factor analysis of the impact of our repair algorithm's design choices on \sysname's FPR in \geant.}
    \label{fig:repair_factor_analysis}
\end{figure*}

\T{Impact of repair on demand validation.} We want to evaluate the impact of the different steps of the repair algorithm on \sys's demand validation. For simplicity, we show results for just the \geant network in this section. We start by assuming a router bug affects 30\% of the counters (denoted \textit{random}) or all counters in 30\% of the routers (denoted \textit{correlated}), and they are either zeroed or scaled down by a random factor chosen uniformly at random in the range $\sbrac{25\%,75\%}$. 

We test scenarios designed to highlight the contribution of certain key design choices: (i)~granting a vote to our demand-induced estimate, $l_{demand}$, (ii)~using multiple rounds of voting, (iii)~the repair step in its entirety.  
\Cref{fig:repair_factor_analysis} illustrates  the resulting FPR.
We see that running the \sys validation without repair yields poor FPR results of over 90\% in all cases. 
If we run a single round of the repair algorithm without the vote from $l_{demand}$, FPR only goes down slightly. However, if we run a single round with all five votes, FPR goes down significantly. Thus, the additional demand-based tie-breaking estimate brings the most significant contribution to FPR. FPR is slightly lower for scaling bugs because they are easier to repair than zeroed counters. With scaling bugs, if two routers that share a link are affected, the two router estimates are different and therefore the error is easier to repair. If both are zeroed, then they agree on the same value and it is harder to make them abandon this value. 
Finally, the multiple voting rounds eliminate most of the remaining FPR by bringing global information in cases where router counters are stuck in local minima. All cases bring down the FPR to under 2\%. Thus, the overall impact on FPR is dramatic, since the network-wide view magnifies it.

\begin{figure}
    \centering
    \includegraphics[width=0.7\linewidth]{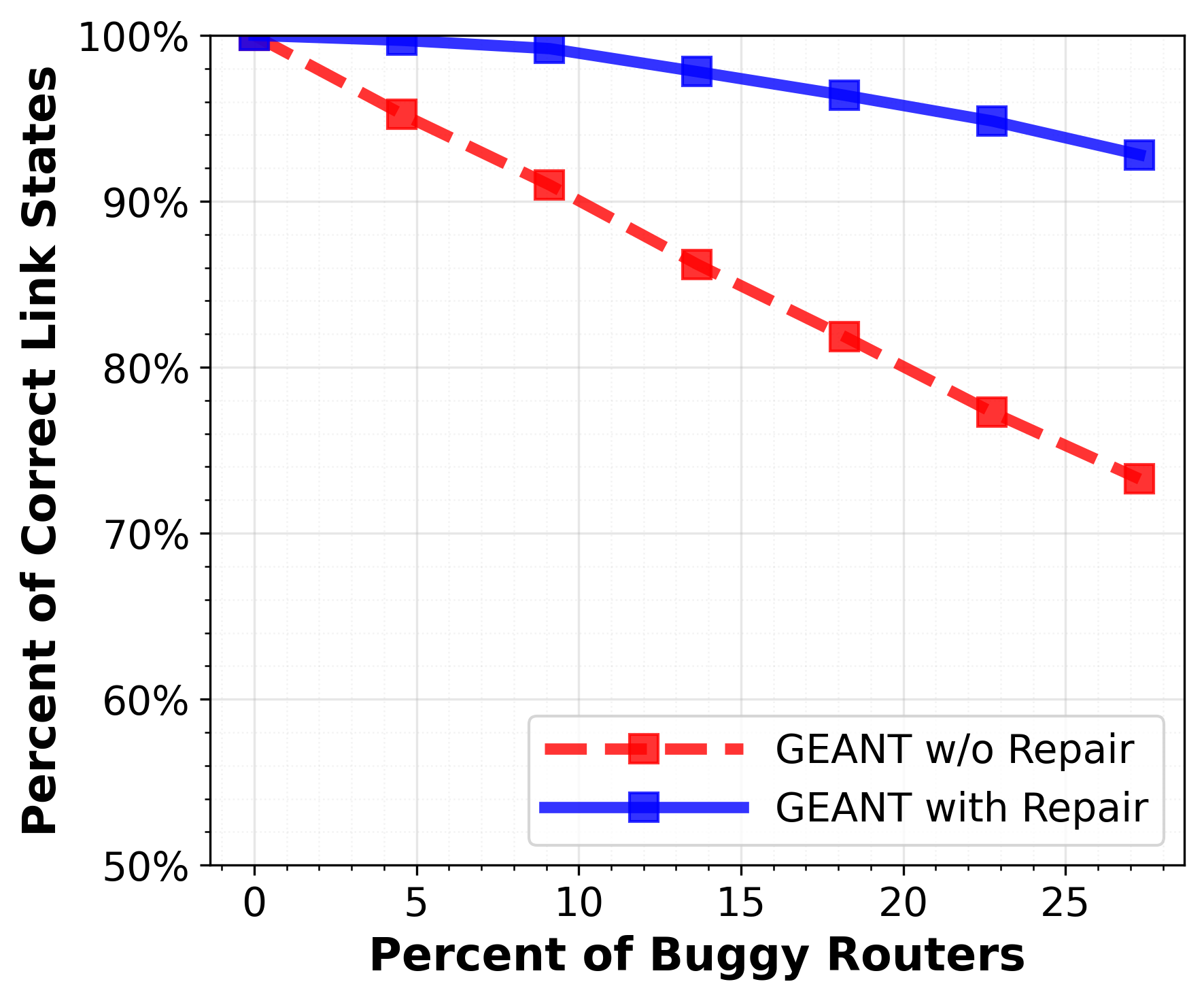}
    \caption{Effectiveness of topology repair in \geant.}
    \label{fig:topology_repair}
\end{figure}

\T{Impact of \sys's repair on topology validation.}  We now want to study the effectiveness of \sys repair for \textit{topology validation} (\cref{sec:design-topology}).
We consider the \geant WAN, and assume a worst case scenario that for any buggy router, \textit{all} telemetry (physical status, link layer status, \textit{and} counters) for all interfaces are buggy, reporting status down and counter 0, even though the links are actually functioning. We run our repair algorithm and plot the percentage of links we can correctly identify to be up, both before and after repair as we increase the number of buggy routers. \Cref{{fig:topology_repair}} shows that repair manages to solve some 2/3 of the incorrect link states, even when over 1/4 of the routers are buggy. 

\medskip

\textbf{In summary}, our evaluation of \sys across both a production deployment and simulation demonstrates that our system is both highly effective at detecting even small (5\%) demand perturbations, and resilient to telemetry causing false positives. Our shadow deployment shows that our design is scalable to production data rates, stable enough to not produce false positives due to real-world noise, yet sensitive enough to catch real-world bugs. 

\vspace{-0.5em}
\section{Related Work}
\label{sec:related-work}
\vspace{-0.5em}

\T{Prior Workshop Publication.}
This paper extends our earlier workshop publication~\cite{hotnets}, which introduced the motivation for input validation in SDN systems and outlined a high-level sketch of a potential solution. In contrast, the design, repair and validation algorithms presented in this paper are entirely new. 
Our prototype implementation and evaluation are also new and were not part of the workshop paper.

\T{Formal verification.} \cameraready{The research community has looked extensively to formal verification as a tool to catch or prevent network issues~\cite{netverify.fun, netkat, bagpipe, tiramisu, minesweeper}. Some works, such as Batfish~\cite{batfish} and HSA~\cite{hsa}, have enjoyed real-world adoption and commercial success. Adjacent approaches include using network emulation directly~\cite{crystalnet, dons}, integrating verification and emulation~\cite{model-free-verif}, or general network testing techniques~\cite{netcastle}. All of these methods aim to enforce some properties on the output behavior given some inputs, without considering the correctness of those inputs. \sysname is complementary to these efforts as it checks at runtime that the inputs correctly reflect reality, which we show is a significant category of outage.}

\T{Output validation.} \cameraready{Some existing tooling considers validating network behavior by comparing actual to expected behavior.} For instance, the Jingubang~\cite{li2024reasoning} module of Hoyan~\cite{hoyan} simulates the expected traffic load of the non-SDN Alibaba WAN, then compares with monitored load. 
However, almost all such approaches aim to prove that an output is correct given a particular input and do not focus on validating the input itself. 

\T{Reasoning language.} 
Flexible Contracts for Resiliency~\cite{sievers2014flexiblecontractsresiliency} provides a language for reasoning about chains of assumptions that connect ground-truth signals to higher-level abstractions, similar to the approach we took in reasoning about the relationship between network signals and controller input. Adopting their formal methods into our system is an interesting direction for future work.

\T{Outage statistics.} The ``Evolve or Die'' paper~\cite{evolve-or-die} extracts lessons from 100+ outages at Google. Although they do include examples of outages due to incorrect inputs, they do not report incorrect inputs as their leading cause of outages.
We speculate that this difference may be due to a few reasons; \eg their analysis looked at both their SDN- and protocol-based WANs, is nearly a decade old, and some of their recommendations may now be common practice leading to new dominant causes. 
Likewise, SkyNet~\cite{skynet} recently reported on the leading network failure root causes at the non-SDN Alibaba WAN.

\T{Anomaly detection.} Anomaly detection~\cite{ahmed2016anomalydetect, patcha2007anomalyoverview} can detect outliers in input data through statistical analysis of a signal's history. In contrast, we focus on whether a signal reflects the ground truth, and so look across signals for corroboration.

\T{Statistical tools.} \cameraready{In the validation step, we essentially want to check whether the distribution of the invariant imbalance is different from a typical distribution given clean input demands, and especially whether the imbalance values are higher. There are many two-sample statistical tests for checking if one distribution with non-negative values is stochastically larger than another, \eg the one-sided Kolmogorov-Smirnov test or the Anderson-Darling test~\cite{kvam2022nonparametric}. Our validation scheme that focuses on the imbalance distribution tail is designed to reduce the sensitivity to bugs in the router counters. Early evaluations indicate that it is competitive with other tests, but more evaluations are needed. }

\vspace{-1em}
\section{Conclusion}
\label{sec:conclusion}

We presented \sysname, a system that validates the inputs to the WAN SDN controller and alerts operators when these inputs appear incorrect. Running \sysname as a shadow validation system in a large WAN, we found that it caught an invalid input on live production data, while maintaining a 0\% false positive rate on healthy inputs, thus avoiding false alarms to the operators. Finally, we showed through simulation that \sysname can catch up to 100\% of demand perturbations over 5\%, while being resilient to noisy telemetry with up to 30\% of telemetry perturbations.

\cameraready{While our focus in this paper is on validating SDN inputs for TE, the methods we present can apply just as well to non-SDN TE such as RSVP-TE, where similar invariants to ours can be checked at each router on the flooded global network state. We further believe this class of input validation problem generalizes to a far wider range of control systems; while the particular invariants we selected for this paper were specific to our problem (\eg relating sums of demand values to interface counters), the methods used to find them were principled and can be applied more broadly to other classes of systems to validate their input, both for network control systems beyond TE such as link health monitoring, and more generally for control systems beyond networks such as building climate control or power management systems.
}

\section*{Acknowledgments}
\cameraready{We thank the anonymous reviewers and our shepherd Matthew Caesar for their insightful comments and helpful feedback. We appreciate early feedback from our colleagues at UC Berkeley, including Scott Shenker, Sarah McClure, Emily Marx, and others. We also thank our colleagues at Google, including Bikash Koley, Brent Stephens, Aniket Pednekar, Anurag Sharma, Hank Levy, David Culler, and others for their discussions and contributions to the development of CrossCheck. We particularly thank Sorin Constantinescu for his deep practical network telemetry knowledge and help in procuring live datasets. This work was partly supported by the Louis and Miriam Benjamin Chair in Computer-Communication Networks.}

\bibliographystyle{plain}
\bibliography{bibliography/nsdi}

\clearpage
\appendix


\section{Link-Invariant Imbalance at WAN $B$}
\label{sec:WAN_B}

\begin{figure*}
  \centering
  \begin{subfigure}[t]{0.35\textwidth}
    \centering
      \hfill
    \includegraphics[width=\textwidth]{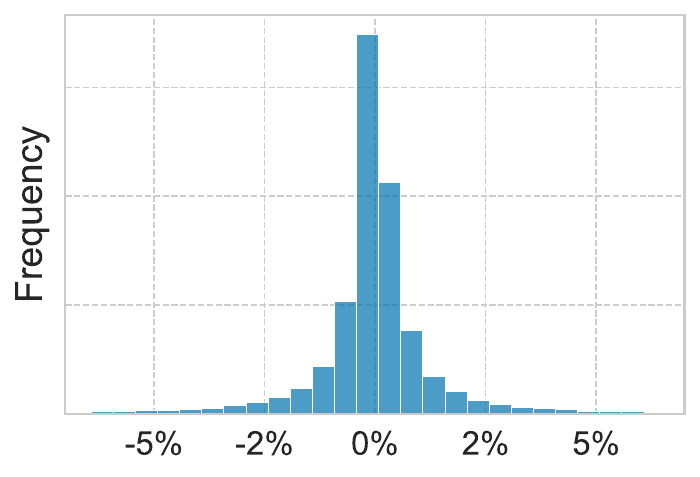}
    \caption{Link invariant over 30 second window.}
    \label{fig:invB:a}
  \end{subfigure}
  \hspace{2cm}
  \begin{subfigure}[t]{0.35\textwidth}
    \centering
    \includegraphics[width=\textwidth]{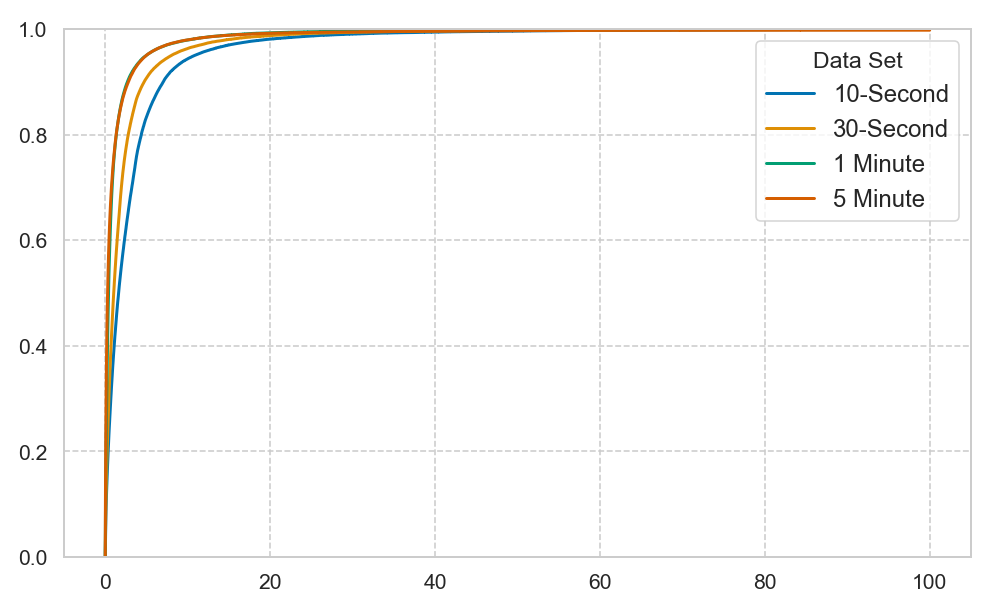}
    \caption{Impact of collection time}
    \label{fig:invB:b}
  \end{subfigure}  
    \hfill
  \caption{Measured imbalance for the link invariants at a large production WAN $B$. 
  }
  \label{fig:invB}
\end{figure*}

\Cref{fig:invB} presents the measured link-invariant imbalance PDF at WAN $B$. First, 
\cref{fig:invB:a} shows how most imbalances hold within 1\%. 
Next, \cref{fig:invB:b} illustrates the impact of collection times on the CDF of the imbalance. By averaging over longer periods, the imbalance decreases, trading off against a longer time before raising a potential alarm. The figure shows how averaging over 1 minute and 5 minutes yield very similar results. 

\eat{
    \T{Mean imbalance in WAN $A$.} We note that in the path-invariant for WAN $A$, we observed a persistent mean imbalance of -2.2\% between demand-path predicted counters and actual counter values. This is due to two factors. First, some router vendors include packet headers in their byte counters and some do not, while demand is instrumented without taking headers into account. Second, there exists some control plane signaling originating at the routers that is not always included in the demand.
    We adjust for this bias in our validation system by expecting a mean imbalance of -2.2\%, \ie automatically substracting this number from the measured values so the average imbalance is zero. 
}
    

\section{Repair Guarantees: Proof of \cref{thm:dtp}}
\label{sec:proof}


\bp 
The theorem is about corrupted counters at a single link. We will distinguish two types of link: \\
\textit{(i)}~\textit{internal links}, which are between two routers of the network, and where two counters can be affected; and \\
\textit{(ii)}~\textit{border links}, which are between a network router and outside the network. Only the counter of the router can be affected.

First, we need to prove that the corrupted counters will not affect other links in the iterative repair computation. For instance, in \cref{fig:DTP-fig},
we want to show that a corruption of the egress and ingress counters at link $X \to Y$ will not affect the link $A \to X$. Assume that the counters at link $j$ are corrupted. If a neighboring link $j'$ is internal, it has 5 estimators, and $j$ only impacts one of these, so the repair algorithm will reach the correct value for $j'$ with confidence of at least $1-\frac{1}{5}=0.8$. Likewise, if $j'$ is a border link, then at least two estimators out of three are unaffected and the algorithm reaches the correct value with confidence of at least $\frac{2}{3}$. Finally, if all neighboring links are unaffected, then a non-neighboring link also cannot be affected, even after several rounds of computation.

Second, we need to prove that the repair algorithm reaches the correct decision for link $j$. We distinguish two cases: \\
\textit{(i)~Internal link:}  Both the demand-based estimator and the router-based estimators are unaffected by the corruption, since the demand vector and the neighboring link counters provide correct values. Thus, the algorithm will correctly converge with a confidence of at least $\frac{3}{5}=0.6$.\\
\textit{(ii)~Border link:} Again, the demand-based and router-based estimators are unaffected, and therefore it will reach the correct decision with a confidence of at least $\frac{2}{3}$. 
\ep

\section{Scaling Model: Proof of \thm{thm:exp}}
\label{sec:proof_exp}

\begin{proof}
Let $p$ (respectively $p'$) denote the probability that the invariant imbalance falls within the threshold of $\tau$ under healthy (resp. buggy) inputs, so $p>p'$. Also let $\cutoff$ denote a fixed validation cutoff, set such that $p>\cutoff>p'$. Then the validation algorithm checks whether after flipping $n$ times a coin with prob. $p$ (resp. $p'$) of getting heads, the proportion of heads is at least $\cutoff$. This is exactly $B_{n,p}(n\cdot \cutoff)$ (resp. $B_{n,p'}(n\cdot \cutoff)$), using the CDF $B_{n,p}$ (resp. $B_{n,p'}$) of the corresponding Binomial distribution.

As a result, $FPR$ and $1-TPR$ have a Chernoff–Hoeffding upper bound that decreases exponentially with $n$~\cite{mulzer2018five,rioul2025information}, which proves the theorem. Namely, 
\begin{equation}
    FPR \leq e^{-n \cdot D(\cutoff \parallel p)}    
\end{equation}
and 
\begin{equation}
    1-TPR \leq e^{-n \cdot D(\cutoff \parallel p')},
\end{equation}
where 
\begin{equation}
D(x \parallel y) = x \cdot \ln\para{\frac{x}{y}}+(1-x)\cdot \ln\para{\frac{1-x}{1-y}}
\end{equation} is the Kullback–Leibler divergence between Bernoulli random variables with parameters $x$ and $y$.
\end{proof}

\section{Repair Pseudocode}
\label{s:dtp-pseudocode}

\SetAlgoNoEnd
\begin{algorithm*}
\SetAlgoLined
\KwIn{Latest measured telemetry values for each interface, stored in \texttt{telemetry\_dict}, \\
\hspace{1.2em} e.g., \texttt{telemetry\_dict}=\{\texttt{r1.eth12}: \{\texttt{send\_rate}: 100, \texttt{recv\_rate}: 75, \texttt{status}: UP\}, \ldots\} \\
Static network layout information, stored in \texttt{network}, \\
\hspace{1.2em} e.g., \texttt{network}=\{\texttt{routers}: \texttt{[r1, r2, $\ldots$}], \texttt{links}: \texttt{[(r1.eth12, r2.eth5), (r1.eth11, r3.eth4), $\ldots$]}\}}

\texttt{locked} $\leftarrow$ \{\} \tcp*{Final results, locked values}

\While{size(\texttt{locked}) $<$ num\_counters}{

    \texttt{possible\_values} $\leftarrow$ \{\}\;
    
    \ForEach{link\_id $\in$ network.links}{ 
        \If{link\_id $\in$ \texttt{locked}}{
            \texttt{possible\_values}[link\_id] $\leftarrow$ [\texttt{locked}[link\_id].value]\;
            \textbf{continue}\;
        }

        src\_intf, dst\_intf = link\_id[0], link\_id[1]\;

        \If{src\_intf is internal}{
            src\_measurement $\leftarrow$ telemetry\_dict[src\_intf][\texttt{send\_rate}]\;
            \texttt{possible\_values}[link\_id].append(src\_measurement)\;
        }

        \If{dst\_intf is internal}{
            dst\_measurement $\leftarrow$ telemetry\_dict[dst\_intf][\texttt{recv\_rate}]\;
            \texttt{possible\_values}[link\_id].append(dst\_measurement)\;
        }
        predicted\_val $\leftarrow$ get\_predicted\_traffic(paths, demand, link\_id)\;
        \texttt{possible\_values}[link\_id].append(predicted\_val)\;
    }

    \texttt{votes} $\leftarrow$ \{\} \tcp*{(src\_intf, dst\_intf): [(v1, weight1), ...]}

    \tcp{Populate votes from router beliefs.}
    \ForEach{router $\in$ network.routers} {
        local\_links $\leftarrow$ get\_links\_attached\_to(router) \tcp*{Any links with src or dst intf at this router}
        
        rtr\_predicted\_vals $\leftarrow$ \{\} \;

        \For{$i \leftarrow 1$ \KwTo $N$}{
            assignment $\leftarrow$ \{\}\;
            \ForEach{link\_id $\in$ local\_links}{
                assignment[link\_id] $\leftarrow$ random.choice(\texttt{possible\_values}[link\_id]) \;
            }
            \ForEach{link\_id $\in$ local\_links}{
                rtr\_predicted\_vals[link\_id].append(($\sum$ assignments) - assignment[link\_id])\;
            }
        }
        \ForEach{link\_id $\in$ local\_links}{
            val, conf $\leftarrow$ most\_frequent\_value\_with\_confidence(rtr\_predicted\_vals[link\_id]) \;
            \texttt{votes}[link\_id].append((val, conf))\;
        } 
    }

    \tcp{Populate votes from possible\_values} 
    \ForEach{link\_id, vals\_list $\in$ \texttt{possible\_values}} {
        \ForEach{val $\in$ vals\_list} {
            \texttt{votes}[link\_id].append((val, 1.0))\;
        }
    }

    assignment\_scores $\leftarrow$ \{\} \tcp*{Lock the value with the highest majority confidence}
    \ForEach{link\_id, votes\_list $\in$ \texttt{votes}}{
        \If{src\_intf $\in$ locked}{
            \textbf{continue}\;
        }
        cumulative\_scores $\leftarrow$ group and sum votes\_list by value with $\leq$ 0.03\% threshold\;
        best\_val, score $\leftarrow$ max(cumulative\_scores)\;
        assignment\_scores[link\_id] $\leftarrow$ (best\_val, score)\;
    }

    link\_id, (best\_val, score) $\leftarrow$ max(assignment\_scores)\;
    \texttt{locked}[link\_id] $\leftarrow$ (best\_val, score)\;
}

\Return \texttt{locked}

\caption{Repair algorithm for telemetry correction.}\label{alg:pseudocde}
\end{algorithm*}

Algorithm~\ref{alg:pseudocde} presents the pseudo-code for the repair algorithm. 

\section{Generating Baseline Link Counters That Match Invariant Noise}\label{sec:gen}

We generate production-realistic telemetry counters for simulation by trying to match the link-invariant, node-invariant and path-invariant noise distributions we observe in the WAN $A$ production data (\cref{fig:inv}). To do this, we first compute the expected amount of traffic on each link with the given demands and paths. 
Then, we perturb the link counters to account for the path-invariant noise. To do so, we add to each link some random i.i.d. noise that is distributed like the path-invariant noise distribution, and assign the noisy link value to its two counters. 
Next, we further modify the interface counters on either side of the link to match the three invariant noise distributions. This is done by first drawing some noise random variable $x$ from the link-invariant noise distribution, then adding $x/2$ to one of the two counters at the end of the link, while removing $x/2$ from the other. Thus, their difference is $x$, following the link invariant noise, and their average has not changed, thus still obeying the path invariant noise.
Finally, the router invariant noise is still unsatisfied. We similarly update counters at each router to obey it, while making sure that the distributions of the other two invariants still hold, until we converge to a satisfying result.
Thus our baseline healthy bug-free telemetry will contain noise matching observed real-world data. 

\section{Additional Evaluations}
\label{sec:more_eval}

\begin{figure}
\centering
\includegraphics[width=0.85\linewidth]{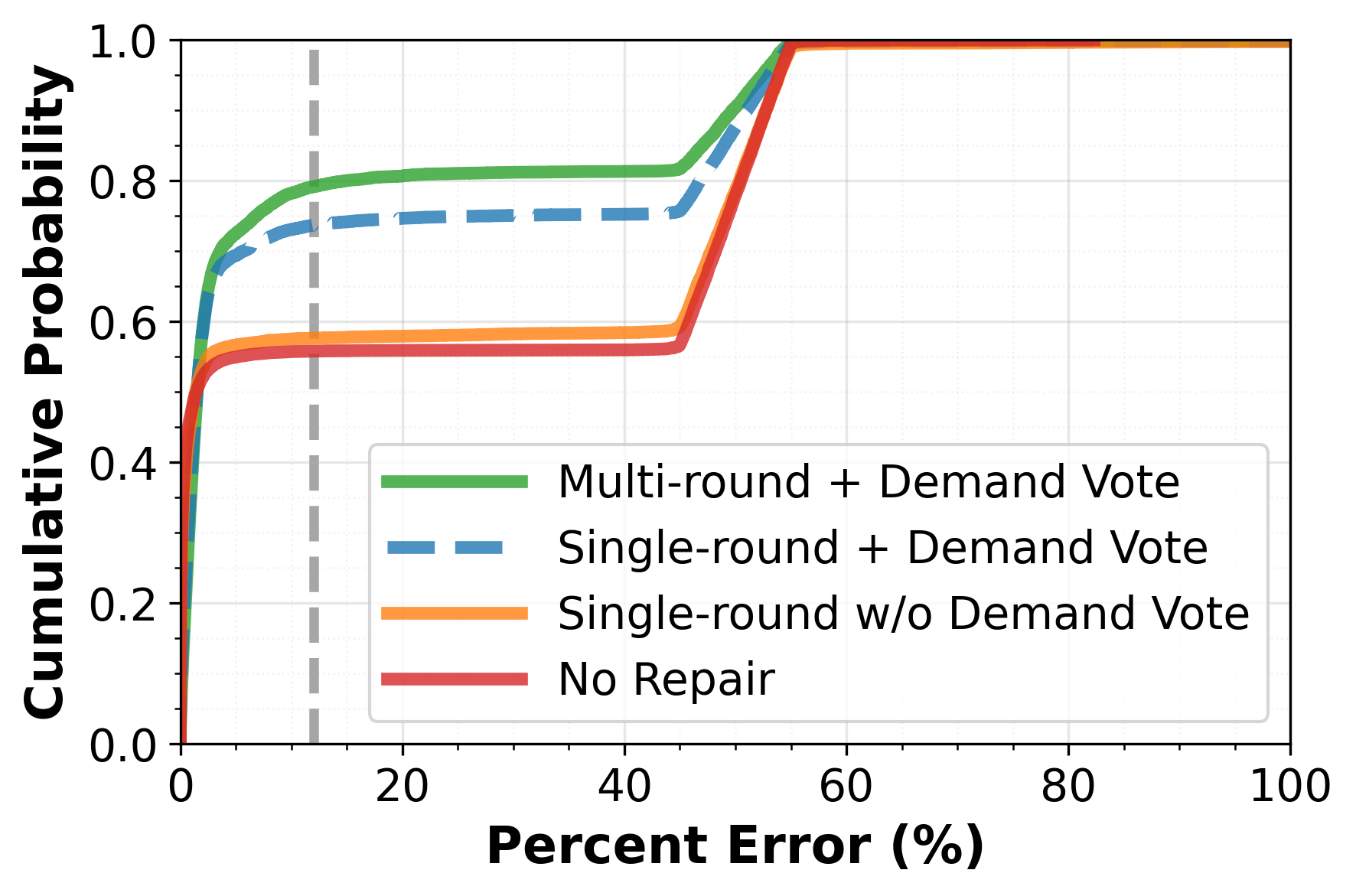}
\caption{Impact of the repair components on the CDF of the counter error in \geant.}
\label{fig:repair_components}
\end{figure}

\T{Do all repair components help reduce counter errors?} 
In the \geant network, we assume that a router bug affects 45\% of the counters, and they are scaled down by a random factor chosen uniformly at random in the range $\sbrac{45\%,55\%}$. \Cref{fig:repair_components} illustrates the contribution of the various components of the repair algorithm on the CDF of the counter error -- \cameraready{we see for the no repair baseline (red solid line), $55\%$ of counters have only noise, and the remaining 45\% have errors between 45\% and 55\%. If we run a single round of the repair algorithm without the fifth vote of the demand-induced link estimate, we only correct an additional 3-4\% of counters. With a single round that includes all five votes, it goes up to 75\% of counters with error under 10\%. Thus, this fifth tie-breaking vote brings the most significant contribution. Finally, if we use the full repair with all gossip rounds, it increases to over 80\% of counters having less than 10\% error, \ie the \sys repair corrected about 2/3 of the router bug-induced errors (from 45\% to 15\%).}

\begin{figure*}
  \centering
  \begin{subfigure}[t]{0.24\textwidth}
    \centering
    \includegraphics[width=\textwidth]{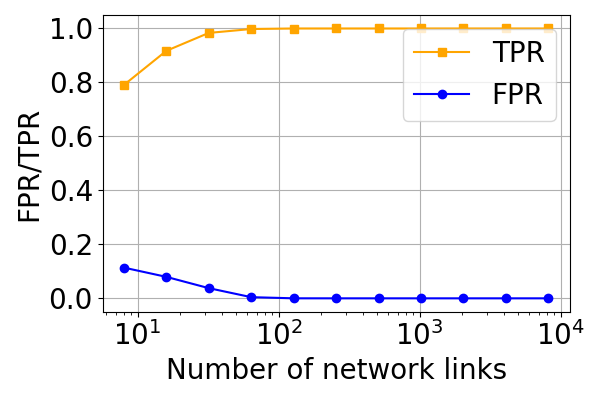}
    \caption{Fixed cutoff}
    \label{fig:n:a}
  \end{subfigure}
  \hfill
    \begin{subfigure}[t]{0.24\textwidth}
    \centering
    \includegraphics[width=\textwidth]{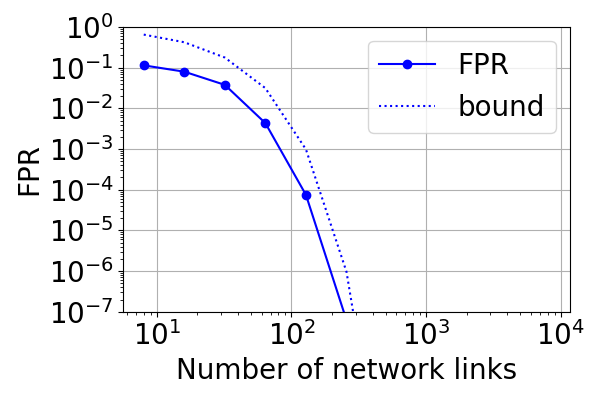}
    \caption{Zoom on $FPR$}
    \label{fig:n:FPR}
  \end{subfigure}
  \hfill
    \begin{subfigure}[t]{0.24\textwidth}
    \centering
    \includegraphics[width=\textwidth]{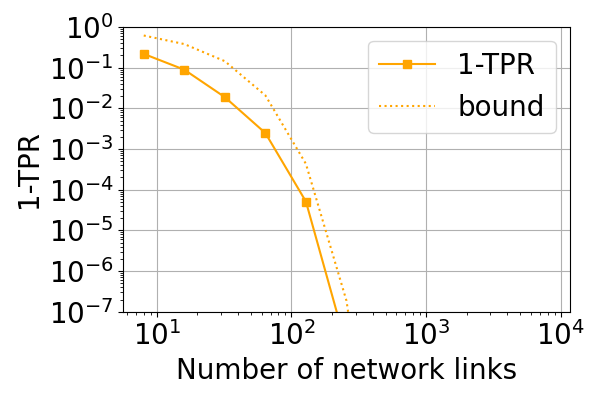}
    \caption{Zoom on $1-TPR$}
    \label{fig:n:TPR}
  \end{subfigure}
  \hfill
  \begin{subfigure}[t]{0.24\textwidth}
    \centering
    \includegraphics[width=\textwidth]{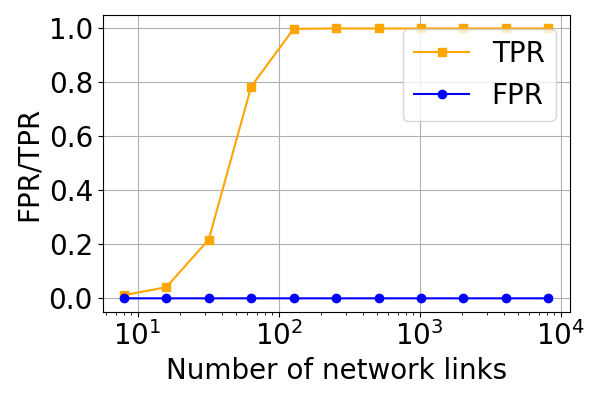}
    \caption{Variable cutoff}
    \label{fig:n:var}
  \end{subfigure}  
  \caption{FPR/TPR scaling model as a function of the number of network links. (a)~Fixed validation cutoff, with exponentially decreasing (b)~$FPR$ and (c)~$1-TPR$. (d)~Tuning the cutoff for each network size to attain a near-zero FPR. 
  }
  \label{fig:n}
\end{figure*}

\T{Does \sysname improve exponentially in larger networks?} \Cref{fig:n} plots the FPR and TPR under the scaling model of \cref{thm:exp}, where the path invariant imbalance distributions are assumed to be the same for all links in all networks.
Specifically, for healthy inputs, the model assumes the path invariant imbalance distribution that was measured in production WAN $A$. 
For buggy inputs, it adds a Gaussian-distributed imbalance to this distribution using $\mathcal{N}(5,5)$. In \cref{fig:n:a}, given a fixed validation cutoff ($\cutoff=0.6$), both TPR and FPR quickly converge to 1 for large networks, but are not optimal for small ones. \Cref{fig:n:FPR,fig:n:TPR} provide a close-up view of $FPR$ and $1-TPR$, illustrating how they decrease exponentially fast with the number of links together with their upper bound, as expected from \cref{thm:exp} and its proof (note that the x axis also uses a log scale). 

\Cref{fig:n:var} sets a different cutoff for each network size, so that $FPR\leq 10^{-6}$, corresponding to at most one false alarm every ten years assuming a validation every five minutes. 
Since there is a TPR vs. FPR tradeoff, TPR becomes worse, and suffers more for small networks. Given that our considered WANs have at least 54 uni-directional links (for Abilene), and modern real-world WANs are much larger, this indicates that \sys should be quite efficient.

\section{Theoretical Aspects of Input Validation}
\label{sec:counter}

\begin{figure}
\centering
\includegraphics[page=1,clip, trim=0cm 3cm 0cm 1cm, width=0.9\linewidth]{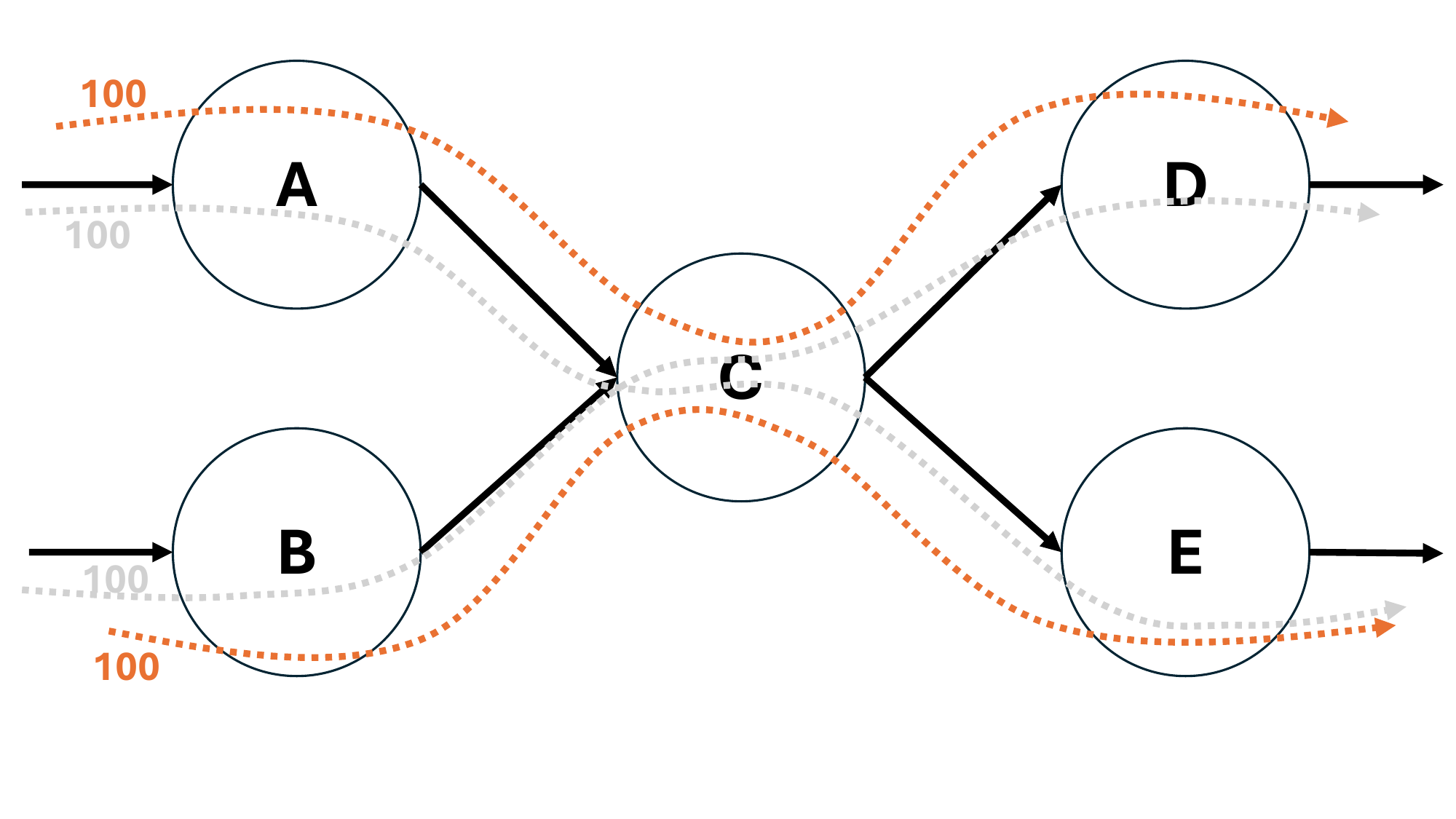}
\caption{Detection counter-example.}
\label{fig:counter}
\end{figure}

\T{Guessing inputs.} The path invariant relates input demands to router counters.
It is natural to ask whether we could guess input 
demands 
given low-level network telemetry. Compressed sensing (CS) and low-complexity iterative message-passing algorithms such as belief propagation (BP) are efficient at inferring the values of a set of random variables, given a set of constraints that induce relations between them using a sparse graph structure, \ie such that each constraint depends on a small number of variables on average~\cite{andrea2008estimating,donoho2009message,firooz2010network}. For instance, Counter Braids~\cite{lu2008counter,lu2009robust} iteratively provides upper and lower bounds on the estimated random variables. These bounds converge to a solution given a sufficient number of constraints.  Furthermore, matrix completion approaches~\cite{keshavan2010matrix,keshavan2010matrix2,candes2012exact} attempt to recover a matrix from a sampling of its entries.
However, our case differs meaningfully. First, the invariants do not suffice to reconstruct the demand matrix, as shown in the counter-example below. Second, we do not know what entries in our corrupted input have been corrupted, and we do not know how they have been corrupted, a model not dealt with in these frameworks. In our case, the bounds provided by the Counter Braids are too wide and miss an overwhelming majority of the data corruption in most corruption scenarios.

\T{Counter-example to guessing demands.} The path invariant asserts that after taking into account the routing, the total demand for a link should equal the router link counters (\cref{eq:path}). 
It is natural to ask whether we could simply reverse-engineer, and guess the input demands from the low-level network telemetry. This would be a simple way to check the input demands.

\fig{fig:counter} provides a counter-example. The demand includes two (source, destination) flows: $(A,D)$ and $(B,E)$. Each flow is of size 100. All of the link counter values will also be 100. However, if a bug reports the demand as consisting of the two pairs $(A,E)$ and $(B,D)$, nothing would change in the link counter values. Thus, both the healthy and buggy demands would yield the same telemetry signals. This illustrates why we cannot unequivocally establish the demand from the lower-level telemetry. 

\end{document}